\documentclass{sig-alternate-05-2015}
\pdfoutput=1
\usepackage{graphicx}
\usepackage{balance}  
\usepackage{epsfig}
\usepackage{times,amsmath}
\usepackage{bm}

\usepackage{tabularx}
\usepackage{mathptmx}
\usepackage{caption}
\DeclareCaptionType{copyrightbox}
\usepackage{subfig}
\usepackage{amsmath}
\usepackage{multirow}
\usepackage{mathtools}
\usepackage[mathscr]{eucal}
\usepackage[square,comma,numbers,sort]{natbib}
\usepackage[linesnumbered,ruled,algonl,vlined,noend]{algorithm2e}
\usepackage{algpseudocode}
\newcommand{\eat}[1]{}

\usepackage[usenames,dvipsnames]{color}
\usepackage{slashbox}
\usepackage{enumitem}
\usepackage{booktabs}
\usepackage[T1]{fontenc}
\usepackage[utf8]{inputenc}
\SetKwRepeat{Do}{do}{while}

\newcommand{\zf}[1]{{\color{Blue} \emph{#1}}}

\newcommand{\mn}[1]{\ensuremath{\mathnormal{#1}}}
\newcommand{\mnth}[1]{\ensuremath{\mathnormal{#1}\cdot}th}

\newcommand\method[1]{{\mbox{\sf\small{#1}}}}

\newcommand{\knn}{\method{$k$NN}\xspace}
\newcommand{\rknn}{\method{R$k$NN}\xspace}

\newcommand{\AP}{\method{AP}\xspace}
\newcommand{\BS}{\method{BS}\xspace}
\newcommand{\Auss}{\method{Aus}\xspace}
\newcommand{\Foursq}{\method{Foursq}\xspace}

\newcommand{\cmpo}{\method{T$m\rho$Q}\xspace}
\newcommand{\tmpo}{\method{T$m\rho$Q}\xspace}
\newcommand{\irr}{\method{IRF}\xspace}

\newcommand{\var}[1]{\mbox{\emph{#1}}}

\def\D{\hphantom{1}}

\newcommand\gb[1]{$#1$\,GB}

\newcommand{\myurl}[1]{{\url{#1}}}

\newcommand{\myparagraph}[1]{\vspace*{-8.4pt}\paragraph*{\normalsize\bf{#1}}}

\newcommand{\noi}{\noindent}
\newcommand{\mycomment}[1]{}

\newcommand{\bigo}{\ensuremath{\mathscr{O}}}

\newcommand{\constraint}{\ensuremath{{\mbox{\em Con}}}}
\newcommand{\dist}{\ensuremath{{\bm{d}}}}
\newcommand{\rank}{\ensuremath{{\bm{r}}}}
\newcommand{\popularity}{\ensuremath{{\bm{\rho}}}}
\newcommand{\res}{\ensuremath{\mathscr{R}}}

\newcommand{\maxdist}{\ensuremath{{\bm{d^{\,\uparrow}}}}}
\newcommand{\mindist}{\ensuremath{{\bm{d^{\,\downarrow}}}}}
\newcommand{\ar}{\ensuremath{{\bm{\hat{r}}}}}
\newcommand{\ap}{\ensuremath{{\bm{\hat{\rho}}}}}

\newcommand{\gain}{\Delta}
\newcommand{\objgain}{\Delta_o}
\newcommand{\blkgain}{\Delta_b}
\newcommand{\maxobjgain}{\ensuremath{\Delta_o^{\,\uparrow}}}
\newcommand{\maxblkgain}{\ensuremath{\Delta_b^{\,\uparrow}}}

\newcommand{\contribution}{\zeta}
\newcommand{\ac}{\zeta}
\newcommand{\popqo}{\ap(o_m,W_{i-1}\backslash qo)}
\newcommand{\popmth}{\ap(o_m,W_i)}
\newcommand{\popmplusoneth}{\ap(o_{m+1},W_{i-1})}
\newcommand{\LR}{\ensuremath{{\bm{r^{\,\downarrow}}}}}
\newcommand{\UR}{\ensuremath{{\bm{r^{\,\uparrow}}}}}

\newcommand{\SR}{\ensuremath{\var{OSR}}} 
\newcommand{\BSR}{\ensuremath{\var{BSR}}}
\newcommand{\OSR}{\ensuremath{\var{OSR}}}

\newcommand{\LBR}{\ensuremath{{\bm{\hat{r}^{\,\downarrow}}}}}
\newcommand{\UBR}{\ensuremath{{\bm{\hat{r}^{\,\uparrow}}}}}

\newlength{\onedigit}
\usepackage{amsthm} 

\newtheorem{example}{Example}
\newtheorem{mydefinition}{Definition}

\newtheorem{lemma}{Lemma}
\newtheorem{corollary}{Corollary}

\begin{document}

\setcopyright{acmcopyright}

\doi{10.475/123_4}

\isbn{123-4567-24-567/08/06}

\conferenceinfo{2016}{}

\acmPrice{\$15.00}

%

\title{Monitoring the Top-{\huge{\em m}} Aggregation in a Sliding Window of
Spatial Queries}

\numberofauthors{8} 
%
\author{
%
%
 Farhana M. Choudhury \qquad  Zhifeng Bao \qquad  J. Shane Culpepper \\
       School of CSIT, RMIT University, Melbourne, Australia\\
       \{farhana.choudhury,zhifeng.bao,shane.culpepper\}@rmit.edu.au\\ \\
       Timos Sellis \\
       Department of CSSE, Swinburne University, Hawthorn, Australia\\
       tsellis@swin.edu.au
}
\date{30 May 2016}

\maketitle

\begin{abstract}
	In this paper, we propose and study the problem of top-$m$
	rank aggregation of spatial objects in streaming queries,
	where, given a set of objects $O$, a stream of spatial
	queries (\knn or range), the goal is to report the $m$
	objects with the highest aggregate rank.
	The rank of an object w.r.t.
	an individual query is computed based on its distance from
	the query location, and the aggregate rank is computed from
	all of the individual rank orderings.
	Solutions to this fundamental problem can be used to monitor
	the importance / popularity of spatial objects, which in turn
	can provide new analytical tools for spatial data.

	Our work draws inspiration from three different domains: rank
	aggregation, continuous queries and spatial databases.
	To the best of our knowledge, there is no prior work that
	considers all three problem domains in a single context.
	Our problem is different from the classical rank aggregation
	problem in the way that the rank of spatial objects are
	dependent on streaming queries whose locations are not known
	a priori, and is different from the problem of continuous
	spatial queries because new query locations can arrive
	in any region, but do not move.

	In order to solve this problem, we show how to upper
	and lower bound the rank of an object for any unseen
	query.
	Then we propose an approximation solution to
	continuously monitor the top-$m$ objects efficiently, for 
	which we design an Inverted Rank File (\irr) index to
	guarantee the error bound of the solution.
	In particular, we propose the notion of {\em safe ranking} to
	determine whether the current result is still valid or not
	when new queries arrive, and propose the notion of {\em
	validation objects} to limit the number of objects to update
	in the top-$m$ results.
	We also propose an exact solution for applications where
	an approximate solution is not sufficient.
	Last, we conduct extensive experiments to verify the
	efficiency and effectiveness of our solutions.
	\end{abstract}

%
%

\ccsdesc[500]{Information systems~Spatial-temporal systems}
\ccsdesc[300]{Information systems~Geographic information systems}
\ccsdesc[300]{Information systems~Data streaming}
%
%

%
%
\printccsdesc


\keywords{spatial indexing; rank aggregation; streaming queries}

\section{Introduction}

Rank aggregation is a classic problem in the database community
which has seen several important advances over the years
~\cite{aggr, aggr_web, TA, TA_new, TA_new_conf, acn08-jacm}.
Informally, {\em rank aggregation} is the problem of combining
two or more rank orderings to produce a single ``best'' ordering.
Typically, this translates into finding the top-$m$ objects
with the highest aggregate rank, where the algorithms used for
ranking and aggregation can take several different forms.
Common ranking and aggregation metrics include majority ranking
(sum, average, median, and quantile), consensus-based ranking (Borda
count), and pairwise disagreement based ranking (Kemeny optimal
aggregation)~\cite{aggr,ranking}.
Rank aggregation has a wide variety of practical applications such as
determining winners in elections, sports analytics, collaborative
filtering, meta-search, and aggregation in database middleware.

One such application area where rank aggregation can be applied 
is in spatial computing~{\cite{sfa16-cacm}}.
In spatial databases for example, a fundamental problem is to rank
objects based on their proximity from a query location.
Range and $k$-nearest neighbor (\knn) queries are two pervasively
used spatial query types.
Given a set of objects $O$ and a query location $q$, a \knn query
returns a ranked list of $k$ objects with the smallest spatial
distance from $q$.
Given a query location $q$ and a query radius $r$, a range query
returns all the objects that are within $r$ distance from $q$,
often sorted by the distance from $q$~\cite{rank_range}.

Spatial queries are an important tool that provides partially ranked
lists over a set of objects.
Each object $o$ receives a different ranking (or is not ranked at
all) which depends on the query location.
Thus, aggregating the ranks of spatial objects can provide key
insights into object importance in many different scenarios.

For example, consider a real estate analytics problem where home
buyers are looking for houses to purchase.
Each person has a preference on housing location, and
a house is ranked based on the distance from a preferred location
(e.g., close to a school or a railway station).
A house that has a high aggregate rank is {\em popular} based on 
two or more users' preferences.
Clearly popularity in this context is a continuous query whose
results change over time as new buyers search for houses, and {\em
recency} can also play an important role when interpreting the final
results.
So, defining ``popularity'' is not immediately obvious in this
example.
However, identifying housing properties with the highest aggregate
rank regardless of how rank is defined is of practical importance for
both buyers and sellers.
The information can either be used to recommend the ``hottest''
houses currently on the market, as a starting search point for a new
buyer, or be used as a metric for a potential seller in monitoring
the ``popularity'' of the houses that are not currently on the market
and help make decisions on when to enter the market.

In this paper, we consider the problem of top-$m$ rank aggregation
of spatial objects for streaming queries, where, given a set of
objects $O$, a stream of spatial queries (\knn or range), the problem
is to report the $m$ objects with the highest aggregate rank.
Here, an object that satisfies the query constraint is ranked based
on its distance from the query location, and the aggregation is computed
using all of the individual rank orderings.
To maintain recency information and minimize memory costs, a sliding
window model is imposed on the query stream, and a query is valid
only while it remains in the window.
We consider one of the most common models for sliding windows, the
{\em count-based} window~{\cite{sliding07}}.

Our work draws inspiration from three different domains -- spatial
databases, continuous queries, and rank aggregation.
While several seminal papers have considered various combinations of
these three domains, no previous work has considered approaches to
combining all three.
We summarize previous work, and the subtle distinctions between
previous best solutions in these problem domains and our work in
Section~\ref{sec:related}.

In the domain of rank aggregation, previous solutions have addressed
the problem of incrementally computing individually ranked lists
using {\em on-demand} algorithms~\cite{aggr_seq,TA_new}.
However, these approaches do not consider streaming queries, and
the best way to extend these approaches to sliding window problem
is not obvious.

In the domain of continuous spatial queries, objects are streaming,
but the queries do not change~{\cite{sliding07,stream_knn,stream_rknn}}.
Continuous result updates of top-$k$ queries where the query location
is changing have also been extensively studied in the
literature~\cite{moving_knn,moving_rknn,moving_safe}.
These approaches make the assumption that a query location can move only
to an adjacent location, and construct a {\em safe region} around the
queries, such that the top-$k$ results do not change as long as the
query location remains in the safe region.
These problems are subtly different from the streaming query problem
explored in this work, where each new query location can be anywhere
in space and the query does not move.

In the domain of spatial databases, other related work on finding the
top objects with the maximum number of Reverse $k$ Nearest Neighbors
(\rknn) exists~{\cite{top_rnn_vldb05,influential_obj,top_rknn_uncertain}}.
Given a set of objects $O$, the $RkNN(o)$ is the set of objects
containing $o$ as a \knn.
Another variant of \rknn is bichromatic, where given a set of objects
$O$ and a set of users $U$, the $RkNN(o)$ is the set of users
regarding $o$ as a \knn of $O$.
Although the count of \rknn is also an aggregation, these solutions
do not consider the rank position of the objects for the aggregation.
Rather, the approaches rely on properties of skyline and $k$-skyband
queries to estimate the number of \rknn for an object.
Finding the exact rank of an object in a skyline or a $k$-skyband is
not straightforward.
Moreover, to the best of our knowledge, there is no previous work on
the continuous case of finding the object with the maximum number of
\rknn for streaming queries (users).

\myparagraph{Our contribution}
In this paper: (i) We propose and formalize the problem of top-$m$
rank aggregation on a sliding window of spatial queries, which draws
inspiration from the three classical problem domains -- rank
aggregation, continuous query and spatial databases.
(ii) We propose an exact solution to continuously monitor the top-$m$
ranked objects.
(iii) We propose an approximation algorithm with guaranteed error
bounds to maximize the reuse of the computations from previous
queries in the current window, and show how to incrementally update
the top-$m$ results only when necessary.
In particular, the following three technical contributions have been
made.
(iv) We propose the notion of {\em safe ranking} to determine whether
the result set in a previous window is still valid or not in the
current window.
(v) We propose the notion of {\em validation objects} which are able
to limit the number of objects to be updated in the result set.
(vi) We show how to use an {\em Inverted Rank File} (\irr) index 
to bound the error of the solution.

To summarize, aggregating spatial object rankings can provide key
insights into the importance of objects in many different problem
domains.
Our proposed solutions are generic and applicable to many different
spatial rank aggregation problems, and a variety of different query
types such as range queries, $k$-nearest neighbor (\knn) queries, and
reverse \knn (\rknn) queries can be adapted and used within our
framework.

The rest of the paper is organized as follows.
Sec.~\ref{sec:related} reviews previous related work.
Sec.~\ref{sec:ps} presents the problem definition and an exact
solution.
Sec.~\ref{sec:rank_bounds} shows how to compute the lower and upper
bounds for rank aggregation using an inverted rank file, which
provides a foundation for an approximate solution to the rank
aggregation problem introduced in Sec.~\ref{sec:approx_method}.
In Sec.~\ref{sec:exp}, we validate our approach experimentally. 
Finally, we conclude and discuss future work in
Sec.~\ref{sec:conclusion}.

\section{Related Work} \label{sec:related}
Since our work draws inspiration from three different problem domains in
database area -- rank aggregation, spatial queries and streaming
queries, we review the related work for each of these problem domains,
and combinations of two domains (if any).

\subsection{Rank aggregation}
Given a set of ranked lists, where objects are ranked in multiple 
lists, the problem is to find the top-$m$ objects with the highest
aggregate rank.
This is a well studied and classic problem, mostly for its importance
in determining winners based on the ranks from different
voters~\cite{Bartholdi1989,TA,TA_new,TA_new_conf,aggr_seq,aggr,aggr_web}.

The approaches of {\citet{aggr}} and {\citet{aggr_web}} assume that
the ranked lists exist before aggregation, and explore exact and
approximate solutions for {\em Kendall optimal aggregation}, and the
related problem of Kemeny optimal aggregation, which is known to be
NP-Hard for $4$ or more lists.
When the complete ranked lists are not available a priori, or random
access in a ranked lists is expensive, {\citet{aggr_seq}} and
{\citet{TA_new}} have shown that the ranked objects for an
individually ranked list can be computed incrementally one-by-one
using on-demand aggregation.
However, these incremental approaches~\cite{aggr_seq,TA_new} are not
straightforward to extend to sliding window models where queries are
also removed from the result set as new queries arrive.




\subsection{Top-k aggregation over streaming data}
A related body of work on aggregation is to find the most frequent items, or finding the majority item over a stream of data~\cite{freq_stream,freq_stream2}. This problem is essentially an aggregation of the count of the data, which is studied for sliding window models as well~\cite{freq_slide, freq_slide2}. The solutions can be categories mainly as sampling-based, counting-based, and hashing-based approaches. The goal of the approaches is to identify the high frequency items and maintain their frequency count as accurately as possible in a limited space. As the result of the queries are not readily available for the count aggregation in our problem, these approaches cannot be directly applied. 

\subsection{Database queries}
The relevant work from the database domain can be categorized mainly
as - (i) moving, (ii) streaming, and (iii) maximum top-$k$.

\myparagraph{Moving queries}
In spatial databases, given a moving query and a set of static
objects, the problem is to report the query result
continuously as the query location
moves~\cite{moving_safe,moving_knn,moving_vstar,moving_rknn,moving_range}.
The most common assumption made to improve efficiency is that 
a query location can move only to a neighboring region~\cite{moving_safe,
moving_vstar,moving_rknn,moving_range}. 
By maintaining a {\em safe region} around the query location,
a result set remains valid as long as the query moves
within that region.
The results must only be updated when the query moves out of the
safe region.
Thus, both the computation and communication cost to report
updated results are reduced.
\citet{moving_knn} substitute the safe region with a set of
{\em safe guarding objects} around the query location such that as
long as the current result objects are closer to query than any
safe guarding objects, the current result remains valid.

The problem of continuously updating \knn results when both the query
location and the object locations can move was initially explored by
\citet{cont_sigmod}.
\citeauthor{cont_sigmod} solve the problem by using a conceptual
partitioning of the space around each query, where the partitions are
processed iteratively to update results when the query or any of the
objects move.
In contrast, streaming queries studied in this paper can originate
anywhere in space and does not move, thus the safe region based
approaches are not applicable.

\myparagraph{Query processing over streaming objects}
Many different streaming query problems have been explored over
the years, among which the problem of continuous maintenance of query
results \cite{stream_knn} is most closely related to our problem.
\citet{stream_knn} explore an expiration time based recency approach
where objects are only valid in a fixed time window.
Other related work explored sliding window models where objects are
valid only when they are contained in the sliding
window~{\cite{sliding07,stream_rknn,topk_sliding_06,topk_pub_sub}}.
The two most common variants of sliding windows are - (i) count
based windows which contain the $|W|$ most recent data objects;
and (ii) time-based windows which contain the objects whose
time-stamps are within $|W|$ most recent time units.
Note that the number of objects that can appear within a time-based
window can vary, when the number of objects in a count based
window are fixed.

The general approach in all of these solutions is as follows --
Queries are registered to an object stream, and as a new object
arrives, the object is reported to the queries if it qualifies as a
result for that query.
The solutions rely on the idea of a skyline, where the set of objects
that are not dominated by any other object in any dimension must be
considered.
In these models, the queries are static, and the skyline is computed
for a query.
Newly arriving objects can be pruned based on the properties of the
skyline.
The key difference between our problem and related streaming problems
is that objects are static in our model while the queries are
streaming.

\myparagraph{Maximizing Reverse Top-$k$}
Another related body of work is reverse top-$k$ querying
~\cite{influential_obj,kmost_demanding, kmost_fav,most_influential}.
Given a set of objects and a set of users, the query is to find the object that is a
top-$k$ object of the maximum number of users.
\citet{influential_obj} explore solutions for spatial
databases using precomputed Voronoi diagrams.
Other solutions for the problem use properties of skyline and
$k$-skyband to estimate the number of users that have an object as a
top-$k$ result.
A $k$-skyband contains the objects that are dominated by at most
$k-1$ objects.
Unlike top-$k$ queries, the number of objects that can be returned by
a range query is not fixed, therefore maintaining a skyband is not
straightforward for range queries.

A related problem in spatial databases is to find a region in space
such that if an object is placed in that region, the object will
have the maximum number of reverse $k$NNs~\cite{maxoverlap, maxfirst,
optregion, film}.
Solutions for this problem depend on static queries (users), and
are therefore not directly applicable to our problem.
Moreover, these solutions do not consider the rank position of the
object in the top-$k$ results in their solutions.

\citet{influential_time} consider the temporal version of the reverse
$k$NN problem.
The score of an object $o$ is defined as the number of \rknn, and
the continuity score of $o$ is defined as the maximum number of
consequent intervals for which $o$ is a top-$m$ highest scored
object.
The goal is to find the object with the highest continuity score.
Although the problem is scoped temporally, both the queries and the
objects in the database are static.

\section{Preliminaries}\label{sec:ps}

\subsection{Problem Definition}
Let $O$ be a set of $N$ objects where $o \in O$ is a single point in
$d$-dimensional Euclidean space, $X^d$. 
Now consider a stream of user queries $\var{SQ}$ which is an infinite
sequence $\langle q_1,q_2,\dots \rangle$ in order of their arrival time.
Each query $q$ is a single point in $X^d$, and associated with a
spatial constraint, $\constraint(q)$, such as range or \knn.
In this work, we focus primarily on range queries, but our solutions
are easily generalized to other spatial query types.


We adopt the sliding window model where queries are a 
continuous ordered stream, and a query is only valid while it
belongs to the sliding window $W$.
We consider only a count-based sliding window in this work,
but time-based windows are also possible.
A count-based window contains the $|W|$ most recent items,
ordered by arrival time.
Before defining our problem, we first present a rank aggregation
measure of an object for a window of $|W|$ queries, denoted as {\em
popularity} which will be used in this work.

\myparagraph{Popularity measure}

Each query $q$ partitions the $O$ objects into two sets such that,
$O^+_q = \{o \in O
\mid o \mbox{ satisfies } \constraint(q) \}$ and $O^-_q = \{o \in O \mid o \mbox{ does not satisfy } \constraint(q) \}$.
Each object $o^+ \in O^+_q$ is ranked based on the Euclidean distance
from $q$, $\dist(o,q)$.
Other distance measures can be used to rank the objects, but are not
considered in this work.
The rank of $o$ with respect to $q$, $\rank(o,q) = i$, is defined as
the \mnth{i} position of $o \in O^+_q$ in an ordered list indexed
from $\mn{i} = 1$ to $|O^+|$ where $\dist(o^+_i,q) \le
\dist(o^+_{i+1},q)$.

The popularity of an object $o \in O$ in a sliding window $W$ of
queries is an aggregation of the ranks of $o$ with respect to the
queries in $W$.
We now formally define {\em Popularity} ($\popularity$) as a rank
aggregation function for a sliding window of $|W|$ queries.
Other similar aggregation functions are applicable to our problem but
beyond the scope of this paper.
$$
\popularity(o,W) = \frac{ \mathlarger{\sum_{i=1}^{|W|}}
   \left \{
\begin{array}{ll}
      N-\rank(o,q_i) + 1 & \mbox{where } o \in O^+_{q_i}\\
      0 & \mbox{otherwise} \\
\end{array}
\right. }
{|W|}
$$


A higher value of $\popularity(o,W)$ indicates higher popularity.
If an object does not satisfy the constraint of a query, the
contribution in the aggregation for that query is zero.

\begin{table}[!ht]
\setlength\tabcolsep{4pt}
\caption{Notation}
\vspace{-8pt}
\hspace{-5mm}
\begin{tabular}{l|p{5.54cm}}
\toprule
{\bfseries Symbol} & {\bf Description}\\
\midrule
\textbf{Section~\ref{sec:ps}} & \\
$W$ & Sliding window of $|W|$ most recent queries. \\
$\dist(o,q)$ & Euclidean distance between object $o$ and query $q$.\\
$\constraint(q)$ & Spatial constraint (range or \knn) of $q$. \\
$\rank(o,q)$ & Ranked position of $o$ based on $\dist(o,q)$.\\
$O^+_q$ & The set of objects in O that satisfy $\constraint(q)$.\\
$\rho(o,W)$ & Popularity (aggregated rank) of $o$ for queries in $W$.\\ 
\textbf{Section~\ref{sec:rank_bounds}} &  \\
$c$ & A leaf level cell of a Quadtree.\\
$\mindist(o,c_q)$ ($\maxdist(o,c_q)$) & The minimum (maximum) Euclidean distance between $o$ and any query in $c_q$.\\
$\LR(o,c_q)$ ($\UR(o,c_q)$) & Lower (upper) bound rank of $o$ for any query $q$ in cell $c_q$.\\
$B$  & Block size of the rank lists. \\
$\mindist(b,c_q)$ ($\maxdist(b,c_q)$) & The minimum (maximum) distance between any object in a block $b$ and any query in cell $c_q$. \\
%
%
\textbf{Section~\ref{sec:partition},~\ref{subsec:approx_measure}}\\
$\epsilon$ & Approximation parameter. \\
$qo$ & The least recent query, which is excluded from $W$. \\
$qn$ & The most recent query, which is added to $W$. \\
$W_{i-1}$, $W_i$ & Two consecutive windows, where $W_i$ is derived from $W_{i-1}$ by excluding $qo$ and adding $qn$. $|W_i|$=$|W_{i-1}|$. \\
$\ar(o,q)$ & Approximate rank of $o$ for $q$.\\
$\ap(o,W_i)$ & Approximate popularity of $o$ for window $W_i$. \\
$\res_i$ & The set of result objects for a window $W_i$. \\
$o_m$ & The \mnth{m} object from the set $\res_{i-1}$.\\
\textbf{Section~\ref{subsec:saferanking}} & \\
$\LBR(b,q)$ ($\UBR(b,q)$) & Lower (upper) bound rank of any object in block $b$ w.r.t. $q$. \\
$\popmplusoneth$ & The approximate popularity of top \mnth{(m+1)} object from $\res_i-1$ in previous window $W_{i-1}$.\\
$\popqo$ & The approximate popularity of top \mnth{m} object from $\res_{i-1}$, updated w.r.t. excluding $q_o$ from Window $W_{i-1}$.\\
$\popmth$ & The approximate popularity of top \mnth{m} object from $\res_i$.\\
\bottomrule
\end{tabular}
\label{table:notation}
\end{table}

Table~{\ref{table:notation}} summarizes the notation used in the
remainder of the paper.
We now formally define our problem as follows: 

\begin{mydefinition}
{\bf Top-$m$ popularity in a sliding
window of spatial queries (\tmpo) problem}.
Given a set of objects $O$, the number of objects to monitor $m$, and
a stream of spatial queries SQ ($q_1,q_2,\dots$), maintain an
aggregate result set $\res$, such that $\res \subseteq O$, $|\res| =
m$, $\forall o \in \res$, $o^{\prime} \in O \backslash
\res$, $\popularity(o, W) \ge \popularity(o^{\prime}, W)$, where $W$
contains the $|W|$ most recent queries.
\end{mydefinition}

\subsection{Baseline}\label{sec:baseline}

A straightforward approach to continuously monitor the top-$m$
popular objects in $W$ is: 
(i) Each time a new query, $qn$ arrives, compute the individual rank
of all the objects in $O^+_{qn}$ that satisfy the query constraint,
$\constraint(qn)$.
(ii) Update $\rho$ of the objects $o \in O^+_{qn}$ for $qn$, and the
objects $o^{\prime} \in O^+_{qo}$ for the query $qo$.
Here, $qo$ is the least recent query that is removed from $W$ as $qn$
arrives.
(iii) Sort all of the objects that are contained in $O^+_q$ for at
least one query $q$ in the current window, and return the top-$m$ objects with the highest $\popularity$ as
$\res$. 
As there is no prior work on aggregating spatial query results in a
sliding window, (See Section~\ref{sec:related}), we consider this
straightforward solution as a baseline approach.

Unfortunately, the baseline approach is computationally expensive for
several reasons:
\begin{enumerate}[leftmargin=1em]
\itemsep 0em
\item For each query, the ranks of all objects that satisfy the
$\constraint(q)$ must be computed.
As the number of objects can be very large, and the queries can arrive
at a high rate, this step incurs a high computational overhead.
\item Each time the sliding window shifts, $\popularity$ for a large
number of objects may need to be updated.
\item The union of all of the objects that satisfy $\constraint(q)$ for
each query in the current window must be sorted by the updated
$\popularity$.
\end{enumerate}

To overcome these limitations, we seek techniques which avoid
processing objects for the query stream that cannot affect the
top-$m$ objects in $\res$.
This minimizes the number of popularity computations that must occur.
Two possible approaches to accomplish this, are: accurately estimate
the rank of the objects for newly arriving queries, or reuse the
computations from prior windows efficiently.
We consider both of these approaches in the following sections.

\section{Rank Bounds and Indexing}\label{sec:rank_bounds}
In this section, we first present how to compute an upper bound and a
lower bound for the rank of an object w.r.t.
an unseen query, and then propose an indexing approach referred to as 
an {\em Inverted Rank File} (\irr) that can be used to estimate the
rank of objects for arriving queries.

\subsection{Computing Rank Bounds} 

  \begin{figure}
    \centering
    \includegraphics[height=1.9in]{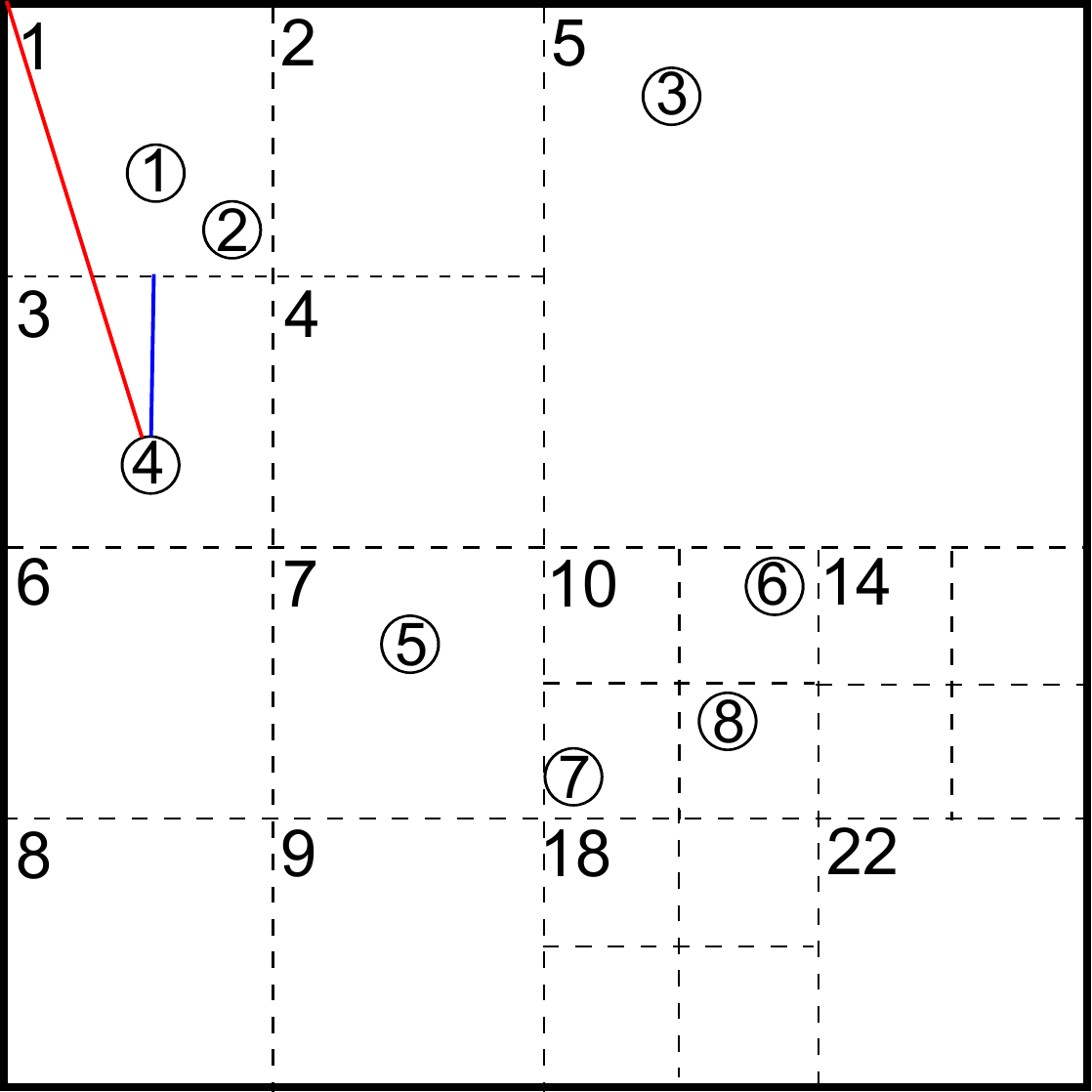}
    \caption{Computing rank bounds}
    \label{fig:rank_bounds}
  \end{figure}

Here, we assume that the space has been partitioned into cells (the
space partitioning step is explained in Section~\ref{sec:partition}).
The rank bound for an object $o$ w.r.t. a cell $c_q$ is 
computed as follows.
For any query $q$ arriving with a location in cell $c_q$, the rank of
$o$ satisfies the condition, $\LR(o,c_q) \le \rank(o,q)
\le \UR(o,c_q)$, where
$\LR(o,c_q)$ and $\UR(o,c_q)$ are the lower and the upper bound
rank of $o$ for any query $q$ in cell $c_q$, respectively.

\myparagraph{Lower rank bound}
The lower rank bound, $\LR(o,c_q)$, is computed such that the rank of
object $o$ will be at least $\LR(o,c_q)$ for a query $q$ contained 
in cell $c_q$.
Note that a smaller value of rank indicates a smaller Euclidean
distance from the query location.
The lower bound rank is computed from the number of objects
$o{^\prime} \in O \backslash o$ that are definitely closer to $q$ in
$c_q$ than $o$. 
%
Specifically, let $\ell_n$ be the number of objects $o{^\prime} \in O
\backslash o$ such that $\maxdist(o{^\prime},c_q) \le \mindist(o,c_q)$,
where $\maxdist(o{^\prime},c_q)$ is the maximum Euclidean distance
between $o^{\prime}$ and cell $c_q$, and $\mindist(o,c_q)$ is the
minimum Euclidean distance between $o$ and $c_q$.
Therefore, even if a query $q$ has a location that is the closest
point of $c_q$ to $o$, there are still at least $\ell_n$
objects closer to $q$ than $o$.
So the rank of $o$ must be greater than $\ell_n$ for any query in
$c_q$, meaning that $\LR(o,c_q) = \ell_n + 1$.
We now give an example of computing the lower rank bound using
Figure~\ref{fig:rank_bounds}.

\begin{example}
Let $O = \{o_1, o_2, \dots, o_8\}$ be the set of objects and $c_{1}$
be a cell in Euclidean space $X^d$.
The minimum distance between $c_1$ and object $o_4$ is shown as the 
blue line.
From Figure~\ref{fig:rank_bounds}, only the maximum distance between
$c_1$ and the object $o_1$ is less than $\mindist(o_4,c_1)$.
Therefore, the lower bound rank of $o_4$ for cell $c_1$ is
\LR$(o_4,c_1) = 1+1 = 2$, i.e., the rank of $o_4$ for any query
appearing in $c_{1}$ must be greater than or equal to $2$.
\end{example}

\myparagraph{Upper rank bound}
The upper rank bound $\UR(o,c_q)$ is the maximum rank an object $o$
can have for any query $q$ appearing in any location of $c_q$.
This bound is computed as the number of objects that can be closer
to $q$ in $c_q$ than $o$.
Let, $u_n$ be the number of objects $o{^\prime} \in O \backslash o$
such that, $\mindist(o{^\prime},c_q) < \maxdist(o,c_q)$ for any query
$q$ in $c_q$, where $\mindist(o{^\prime},c_q)$ is the minimum
distance between $o{^\prime}$ and $c_q$ and $\maxdist(o,c_q)$ is the
maximum distance between $o$ and $c_q$.
Therefore, even if a query $q$ arrives at the farthest location in
$c_q$ from $o$, there are at most $u_n$ objects that can possibly be
closer to $q$ than $o$.
So the rank of $o$ cannot be greater than $u_n+1$ for any query in
$c_q$, resulting in $\UR(o,c_q) = u_n+1$.

\begin{example}
In Figure~\ref{fig:rank_bounds}, the maximum distance from $o_4$ to
cell $c_1$ is shown with a red line.
Here, the minimum distance between $c_1$ and each of the objects
$o_1,o_2$ and $o_3$ is less than $\maxdist(o_4,c_1)$.
So, \UR$(o_4,c_1) = 3+1 = 4$.
Therefore, the rank of $o_4$ for any query in $c_{1}$ must be less
than or equal to $4$.
\end{example}

\subsection{Indexing Rank}\label{sec:index}
\begin{figure}
\centering
\includegraphics[height=1.1in]{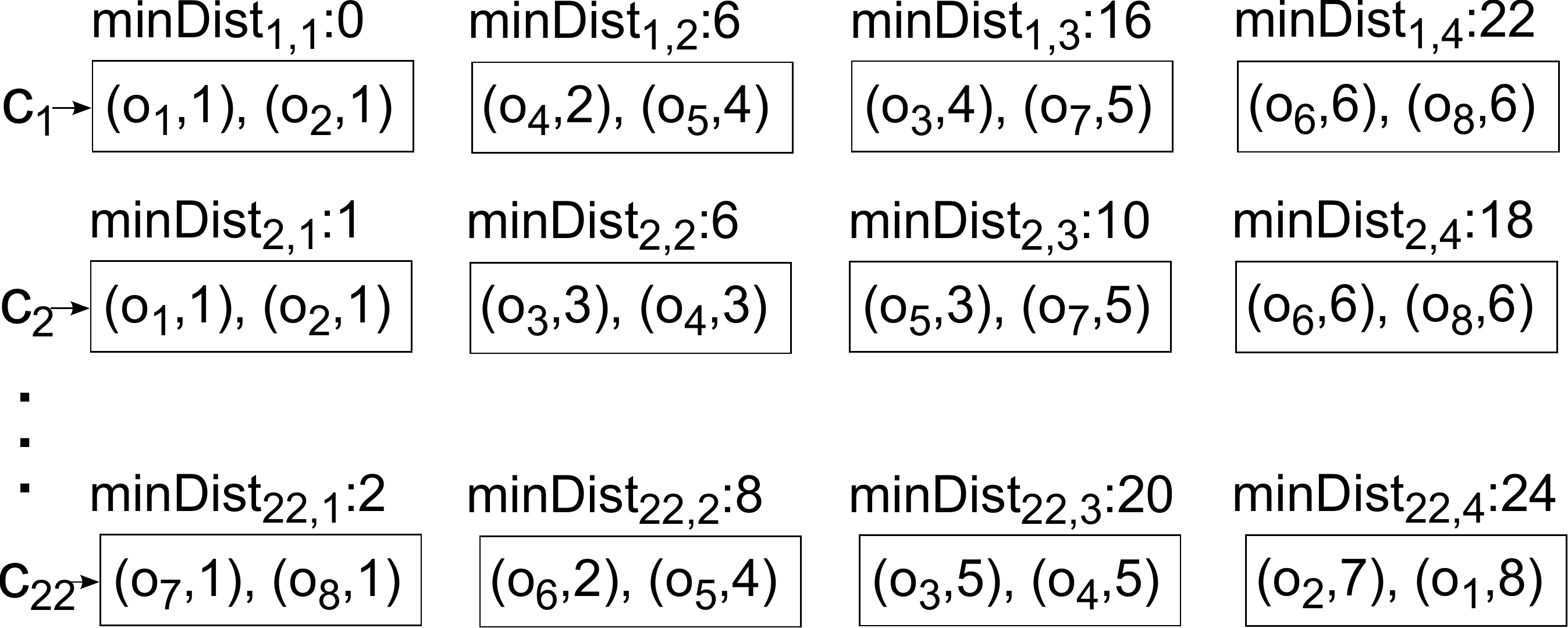}
\vspace{-6pt}
\caption{An example inverted rank file}
\label{fig:index}
\end{figure}

We present an indexing technique called an {\em Inverted Rank
File} (\irr) where $X^d$ is partitioned into different cells, and the
rank bounds of each object for queries appearing in the cell is
precomputed.
The rank information is indexed such that, if a query $q$ arrives
anywhere inside a cell $c_q$, the rank of any object for $q$ can be
estimated.
A quadtree structure is employed to partition $X^d$ into cells.

First, we present the general structure of an \irr.
Later in Section~\ref{sec:partition} we present a space partitioning
approach to approximately answer \cmpo with a guaranteed error bound, and present the rationale behind using a quadtree for space
partitioning.

\myparagraph{Inverted Rank File}
An inverted rank file consists of two components, a collection of all
leaf level cells of the quadtree, and a set of \emph{rank lists},
one for each leaf level cell $c$ of the quadtree.
Each rank list is a sorted sequence of tuples of the form $\langle o,
\LR(o,c_q)\rangle$, one for each object $o \in O$, sorted in ascending
order of $\LR(o,c_q)$.
If multiple objects have the same $\LR(o,c_q)$ for a cell $c_q$,
those tuples are sorted by $\mindist(o,c_q)$.
Here, $\LR(o,c_q)$ is the lower bound rank of $o$ for a query $q$
coming in cell $c_q$.
Note that, for any object $o$, $\LR(o,c_q) \le \rank (o,q)$ holds for
any query $q$ arriving in any location of $c_q$.
Each rank list is stored as a sequence of blocks of a fixed length,
$B$.
Each block $b$ of the rank list for cell $c_q$ is associated with the
minimum distance between $c_q$ and any object in $b$, where 
$\mindist(b,c_q) = \min_{o \in b} \mindist(o,c_q)$.

\begin{example}
Figure~\ref{fig:index} illustrates an \irr index for the objects $O
= \{o_1, o_2, \dots, o_8\}$ shown in Figure~\ref{fig:rank_bounds}.
Assume that the quadtree partitions the space into $22$ disjoint leaf
cells as shown in Figure~\ref{fig:rank_bounds}, and the size of
each block is $B = 2$.
As a specific example, the lower bound rank of the object $o_4$ is
$2$ for the quadtree cell $c_1$.
If a new query arrives in any location contained within cell $c_1$,
the lower bound of the rank of $o_4$ is $2$.
\end{example}

\section{Approximate Solution}\label{sec:approx_method}

As the rate of incoming queries can be very high, there may be
instances where the cost of computing the exact solution is too 
expensive.
In this section, an approximate solution for the {\cmpo} problem is
presented.
This can be accomplished by using the rank bounds to create an
approximate solution with a guaranteed error bound.
At the highest level, the approximate solution consists of the
following steps: 

\begin{enumerate}[leftmargin=1em]
\itemsep 0em
\item A space partitioning technique is used to construct an \irr
index in order to support the incremental computation of the
approximate solution of {\cmpo}.
\item A {\em safe rank} is computed which represents a threshold, if this threshold is exceeded by a result object currently in $\res$, that object
must remain in $\res$ as a valid result.
Specifically, the current safe rank can be computed by combining: (i)
a block based safe rank; and (ii) an object based safe rank.
\item If the ranks of all the result objects are \emph{safe}, $\res$
does not need to be updated.
Otherwise, more work must be done to determine if any object can 
affect $\res$.
This can be achieved using a second technique called {\em validation
objects}, which incrementally identifies the objects than can affect
$\res$.
As long as the current result objects have a higher popularity than
the validation objects, $\res$ does not need to be updated.
\item If $\res$ must be updated, the approximate popularity of the
affected objects are computed.
We show that the popularity computations of the prior windows can
be used to efficiently approximate the popularity scores of the
objects that must change.
\end{enumerate}

We first present the space partitioning approach used to construct the
{\irr} for approximate results in Section~\ref{sec:partition}.
Then, we present the approximate popularity measure of an object
using rank bounding in Section~\ref{subsec:approx_measure}.
In particular, we first outline the workflow of our approximation
algorithm, then we propose the notion of {\em safe ranking} to
determine whether the current result is still valid or not (in
Section~\ref{subsec:saferanking}), and the notion of {\em validation
objects} which limit the number of objects to update in the result
whenever the window shifts (in Section~\ref{subsec:vo}).
Finally, in Section~\ref{sec:error} we discuss the error bound
guarantees provided by our approximation algorithm.

\subsection{Space partitioning}\label{sec:partition}

Ideally, rank bound estimations should be as close as possible to
the actual rank of each object.
If the quadtree leaf cell where a query $q$ arrives is as small as
a single point location (i.e., the same as the location of $q$), then
both the upper and the lower bound ranks of any object will be
exactly the same as the actual rank of that object for $q$.
However, if the space is partitioned in this way, then each point in
$X^d$ will become a leaf cell of the quadtree, and the number of
cells will be infinite.
Therefore, we propose a partitioning technique which guarantees that 
the difference between the rank bounds of any object and its true 
rank is bounded by a threshold, $\epsilon$.

Specifically, for any $o \in O$, and any leaf level cell $c$ of the
quadtree, the difference between the upper and the lower bound rank must
be within a {\em percentage of the lower bound rank}:
\begin{equation}
\UR(o,c) - \LR(o,c) \le \epsilon \times \LR(o,c) 
\end{equation}
\noi Otherwise, cell $c$ is further partitioned until the condition holds. 
As an example, let the threshold be $50\%$ of the lower bound,
$\epsilon = 0.5$.
For an object $o$, and a cell $c_i$, let $\LR(o,c_i) = 10$ and
$\UR(o,c_i) = 20$.
So, the cell $c_i$ needs to be further partitioned for $o$ until the
condition is met.
As another example, for the same object $o$ and another cell $c_j$,
let $\LR(o,c_j) = 100$ and $\UR(o,c_j) = 120$.
Now cell $c_j$ does not need to be partitioned for $o$ since $120 - 100 \le 0.5 \times 100$.

The intuition behind this partitioning scheme becomes quite clear when
the notion of ``top'' ranked objects is taken into consideration.
Getting the exact position of the highest ranked object matters much
more than getting the exact position of the object at the thousandth
position.
So, the granularity of exactness in our inequality degrades gracefully
with the true rank of the object.

\setlength{\algomargin}{1.2em} 
\begin{algorithm}
\caption{{\sc Quadtree Partition}($O, \epsilon$)}
\label{algo:partition}
\begin{small}
Initialize {\sf Quadtree} with the $X^d$\\
{\em node} $\leftarrow$ {\sf Quadtree}({\em root})\\
{\sf Quadtree}({\em root}) $\leftarrow$ {\sc Partition}({\em node}, $O,\epsilon $)\\
{\sc return} {\sf Quadtree}\\ 
\hrulefill \\
{\sc Procedure Partition}($\var{node}$,$O,\epsilon $)\\
$O^\prime$ $\leftarrow \varnothing$\\
\For{$ o \in O$}
{	
	\uIf{$\UR(o,\mbox{node}) - \LR(o,\mbox{node}) > \epsilon \times \LR(o,\mbox{node})$}
	{
		$O^\prime \leftarrow o$\\
	}
}
\uIf{$O^\prime \ne \varnothing$}
{
	{\sc Split}($\var{node}$)\\
	\For{child of node}
	{
		$\var{child}$ $\leftarrow$ {\sc Partition}($\var{child}$, $O^\prime, \epsilon$)\\
	}
}
{\sc return} {\em node}\\
{\sc end procedure}
\end{small}
\end{algorithm}

\myparagraph{Partitioning process}
The partitioning of $X^d$ using this strategy can be achieved
iteratively.
Algorithm~\ref{algo:partition} illustrates the 
partitioning process.
The root of the quadtree is initialized with the entire space $X^d$.
The process starts from the root cell and recursively partitions
$X^d$.
If the partitioning condition is not satisfied for an object $o$ and
a cell $c$, partitioning of $c$ continues until Condition (1) is met
(Lines 1.8 - 1.14).
The process terminates when for each object $o$, the partitioning
condition holds for all of the current leaf level quadtree cells $c$.

\medskip
\noindent{\bf Why use a Quadtree?}
{\xspace} We use a quadtree to partition the space and then
organize the spatial information for each quadtree cell.
The rationale for using a quadtree is as follows: (i) The quadtree
partitions the space into mutually-exclusive cells.
In contrast, MBRs in an R-tree may have overlaps, so a query
location can overlap with multiple partitions, making it 
difficult to estimate the object ranks in new queries.
(ii) A quadtree is an update-friendly structure, and 
the partitioning granularity can be dynamically changed
using $\epsilon$ to improve the accuracy bounds.
This allows performance to be quickly and easily tuned
for different collections.
(iii) In a quadtree, a cell $c$ is partitioned only when any rank
bounds for $c$ do not satisfy Condition (1).
In contrast, if a regular grid structure of equal cell size is used,
enforcing partitioning using Condition (1) will result in
unnecessary cells being created.

Now we present the approximate popularity measure for an object in a
sliding window $W$ of queries using the new rank bounds.

\subsection{Framework of approximate solution}\label{subsec:approx_measure}
In this section, we first introduce how to compute the approximate
popularity of an object for a given sliding window, then we 
show how to aggregate the top-$m$ approximate results.
Since this section is all about how to compute the approximate
popularity of objects, we use the terms popularity and approximate
popularity interchangeable, unless specified otherwise.

First, a lemma is presented to show that the rank of any object
$o$ for a query $q$ arriving in a cell $c$ can be estimated using
only the lower bound rank, $\LR(o,c)$ within an error bound.

\begin{lemma}\label{lemma:rank_bound}
For any object $o \in O$, and any query $q$ arriving in cell $c$, $\LR(o,c)
\le \rank(o,q)
\le (1+\epsilon) \times \LR(o,c)$ always holds.
\end{lemma}
\begin{proof}
The rank bounds are computed such that $\LR(o,c) \le \rank(o,q) \le
\UR(o,c)$ always holds.
For any object $o \in O$, and for any leaf level cell $c$ of the
quadtree, the space is partitioned in a way that guarantees $\UR(o,c)
- \LR(o,c) \le \epsilon \times \LR(o,c)$, so, clearly $\LR(o,c) \le
\rank(o,q) \le (1+\epsilon) \times \LR(o,c)$ also holds.
\end{proof}

Based on Lemma~\ref{lemma:rank_bound}, we approximate the rank of an
object with an error bound as: 
\begin{equation} \label{eqn:ar}
\ar(o,q)  = (1+ \frac{\displaystyle \epsilon}{\displaystyle 2}) \times \LR(o,c) 
\end{equation}

\begin{corollary}\label{lemma:rank_error1}
For any object $o \in O$, and any query $q$ arriving in cell $c$,
$|\rank(o,q) - \ar(o,q)| \le \epsilon/2 \times \LR(o,c)$ always holds.
\end{corollary}
\begin{proof}
Here, $\epsilon/2 \times \LR(o,c)$ is the average of
$\LR(o,c)$ and $(1+\epsilon) \times \LR(o,c)$.
Therefore, the proof follows from Lemma~\ref{lemma:rank_bound}.
\end{proof}

The approximate popularity $\ap(o,W)$ of an object $o$ for the
queries $q$ in a $W$ can be
computed using rank approximation as: 
\begin{equation} \label{eqn:ap}
\ap(o,W) = \frac{ \mathlarger{\sum_{i=1}^{|W|}}
   \left \{
\begin{array}{ll}
      N-\ar(o,q_i)+1 & \mbox{where } o \in O^+_{q_i} \\
      0 & \mbox{otherwise} \\
\end{array}
\right. }
{|W|}
\end{equation}

Next we present the algorithm to compute the top-$m$ objects with the
highest approximate popularity in the sliding window.
Later in Section~\ref{sec:error}, we show how to bound the approximation
error.

Updating a count based sliding window $W_i$ of queries from the
previous window $W_{i-1}$ can be formulated as replacing the least
recent query $qo$ by the most recent query $qn$ when the sliding
window shifts.
As a result, only the leaf level cells (in the quadtree) that contain
$qn$ and $qo$ need to be found, namely $c_{qn}$ and $c_{qo}$.
The rank lists corresponding to these cells can be quickly retrieved
from the \irr index.
For each window $W_i$, assume that $m+1$ objects with the highest
$\ap$ are computed, where the top $m$ objects are returned as the
result $\res_i$ of \cmpo for $W_i$, and the popularity of the
\mnth{(m+1)} object is used in the next window to identify the safe
rank and the validation objects efficiently.

\setlength{\algomargin}{1.2em} 
\begin{algorithm}[t]
\caption{{\sc \cmpo} 
}
\label{algo:approx}
\begin{small}
{\bf Input:} \\
\hspace{2mm} Window $W_i$, number of result objects $m$,
the result objects $\res_{i-1}$ of \\
\hspace{2mm} the previous window $W_{i-1}$, and
the \mnth{m+1} best popularity \\
\hspace{2mm} $\popmplusoneth$ of the previous window $W_{i-1}$.\\
{\bf Output:}~~Result objects $\res_i$ of the current window $W_i$.\\
Initialize a max-priority queue $PQ$\\
$\res_i \leftarrow \emptyset$\\
$qn \leftarrow W_i \backslash W_{i-1}$\\
$qo \leftarrow W_{i-1} \backslash W_{i}$\\
\For{$o \in \res_{i-1}$}
{
	 $\ap(o,W_{i-1}\backslash qo) \leftarrow \ap(o,W_{i-1}) - \frac{\displaystyle \contribution(\ar(o,qo))}{\displaystyle |W_{i}|}$
}
$\ap(o_m,W_{i-1}\backslash qo) \leftarrow$ the approximate popularity of top \mnth{m} object from $\res_{i-1}$ after updating for $qo$.\\

$\var{BSR} \leftarrow$ {\sc Block\_safe\_rank}($\popmplusoneth$,$\ap(o_m,W_{i-1}\backslash qo)$,$PQ$)\\
\For{$o \in \res_{i-1}$}
{
	$\ap(o,W_i) \leftarrow \ap(o,W_{i-1}\backslash qo) + \frac{\displaystyle \contribution(\ar(o,qn))}{\displaystyle |W|}$\\
	\uIf{$\ar(o,qn) \le \mbox{BSR}$ {\sc and} $o \in O^+_{qn} $}
	{
		$\res_i \leftarrow o$
	}
}

\uIf{$|\res_i|<m$}
{
	$\popmth \leftarrow$ current \mnth{m} best popularity of $\res_{i-1}$ in $W_i$. \\
	$\var{OSR} \leftarrow$ {\sc Object\_safe\_rank}($\popmplusoneth, \popmth, PQ$)\\
	\For{$o \in \res_{i-1} \backslash \res_i $}
	{
		\uIf{$\ar(o,qn) \le \mbox{OSR}$ {\sc and} $o \in O^+_{qn} $}
		{
			$\res_i \leftarrow o$
		}
	}
}

\uIf{$|\res_i|<m$}
{
	\var{VO} $\leftarrow$  {\sc Validation\_objects}($\popmth$, $PQ$)\\
	\uIf{VO $\ne \emptyset$}
	{	
		$\res_i$ $\leftarrow$ {\sc Update\_results}($\var{VO}, \res_{i-1} \backslash \res_i$)\\
	}	
}

{\sc return} $\res_i$
\end{small}
\end{algorithm}

The steps for updating the approximate solution of \cmpo for a window
$W_i$ are shown in Algorithm~\ref{algo:approx}.
Note that notation was previously defined in
Table~\ref{table:notation}.
Here, $\contribution(o,q)$ is the contribution of $q$ to the
popularity of $o$, and is computed as: 

$
\contribution(\ar(o,q)) = \left \{
\begin{array}{ll}
      N-\ar(o,q) + 1 & \mbox{where } o \in O^+_q\\
      0 & \mbox{otherwise} \\
\end{array}
\right. 
$

First, the approximate popularity of the result objects $o \in
\res_{i-1}$ for the excluded query $qo$ with $\ar(o,qo)$
is updated.
Let the updated \mnth{m} highest popularity from the set of
$\res_{i-1}$ be $\ap(o_m,W_{i-1}\backslash qo)$ (Lines 2.10 - 2.12 in
Algorithm~\ref{algo:approx}). 
The rest of the algorithm consists of three main components - (i)
computing the safe rank in two steps (block based and object based
safe rank), (ii) finding the set of validation objects, and (iii)
updating $\res_i$.

\myparagraph{Locating an object in {\irr}}
Since some of the steps in Algorithm~\ref{algo:approx} require
finding the entry of a particular object in a rank list in \irr, we
first present an efficient technique for locating objects,
and then describe the remaining steps of the approximation
algorithm.

We start with a lemma to find a relation between the minimum Euclidean
distance of the objects from a cell $c$ and the lower rank bounds of
the objects for any query in cell $c$.

\begin{lemma} \label{lemma:sorting}
For any two objects $o_i, o_j \in O$ and a cell $c$, if $\LR(o_i,c)
\le \LR(o_j,c)$, then $\mindist(o_i,c) \le \mindist(o_j,c)$ always holds.
\end{lemma}
\begin{proof}
We prove the lemma using proof by contradiction.
Assume that $\mindist(o_i,c) > \mindist(o_j,c)$ is true.
In the rank list of \irr, the entries $\langle o, \LR(o,c)\rangle$ are
sorted in ascending order of the lower rank bound, $\LR(o,c)$.
Here, $\LR(o,c)$ is the number of objects $o^\prime$ (plus 1) that
are guaranteed to be closer to $c$ than $o$.
So, $\maxdist(o^\prime, c) \le \mindist(o,c)$.
Let $\mid O^\prime_i \mid = \ell_i$ be the set of all the
objects from $O$ such that $\forall o^\prime_i \in O^\prime_i$,
$\maxdist(o^\prime_i, c) \le \mindist(o_i,c)$, and $\mid O^\prime_j
\mid = \ell_j$ be the set of objects where $\forall o^\prime_j \in
O^\prime_j$, $\maxdist(o^\prime_j, c) \le \mindist(o_j,c)$.
As $\LR(o_i,c) \le \LR(o_j,c)$, then $\ell_i \le \ell_j$ is also true.

Since $\mindist(o_i,c) > \mindist(o_j,c)$ was assumed to be true,
$\maxdist(o^\prime_j, c) \le \mindist(o_j,c) < \mindist(o_i,c)$.
Therefore, $o^\prime_i$ and $o^\prime_j$ both are in the set of
objects from $O$ that satisfy $\maxdist(o^\prime_i, c) \le
\mindist(o_i,c)$, and $\maxdist(o^\prime_j, c) < \mindist(o_i,c)$,
respectively.
Hence, $O^\prime_j \subseteq O^\prime_i$, so $\ell_j \le \ell_i$
must be true.
But this is a contradiction.
Therefore, $\mindist(o_i,c) > \mindist(o_j,c)$ cannot be true.
If $\LR(o_i,c) \le \LR(o_j,c)$ is true, $\mindist(o_i,c)
\le \mindist(o_j,c)$ must hold.
\end{proof}

Lemma~\ref{lemma:sorting} show that sorting the objects by the value
of $\LR(o,c)$ is equivalent to sorting the objects by their minimum
Euclidean distance to $c$, $\mindist(o,c)$.
If multiple objects have the same $\LR(o,c)$, they are already stored
as sorted by their $\mindist(o,c)$ as described in Sec.~{\ref{sec:index}}.
Therefore, we can locate an entry position of object $o$ in the rank
list of cell $c$ (in {\irr}) in three steps.\\ (1) Compute the
minimum Euclidean distance $\mindist(o,c)$ of $c$ from $o$.\\ (2)
Using an {\irr} as described in Sec.~\ref{sec:index}, each block $b$
of the rank list for cell $c$ is associated with the minimum distance
between $c_q$ and any object in $b$, $\mindist(b,c)$, so a binary
search on $\mindist(b,c)$ can be performed to find the position of
the block $b$ where $o$ is stored.\\ (3) Perform a linear scan in
that block to find the entry for $o$. 

\smallskip

The entire process has $\bigo(log_2(N/B) + B)$ time complexity, where
$B$ is the number of objects in a block.

\subsection{Safe rank}\label{subsec:saferanking}
Recall that in Algorithm~\ref{algo:approx} the purpose of finding a
safe rank is to minimize the number of updates in $\res_{i-1}$
(result objects in the previous window $W_{i-1}$) to get the result set
$\res_i$ (in the current window $W_i$) whenever the sliding window
shifts.
In particular, the idea is to compute the safe rank $\SR$ for the objects
$o \in \res_{i-1}$ such that, if $\ar(o,qn) < \SR$, then no other
object from $o^{\prime} \in O \backslash \res_{i-1}$ can have a
higher $\ap$ than $o$, thereby $o$ is a valid result in $\res_{i}$ as
well.
Note that a smaller value of rank implies a higher contribution in
the popularity measure.
The safe rank is defined w.r.t.
the current window $W_i$ by default.

Before presenting the computation of an object's safe rank,
the concept of {\bf popularity gain} of an object $o$ is introduced, which
results from replacing the least recent query $qo$ by the most recent
query $qn$, and is denoted by $\objgain$:
\begin{equation}\label{eqn:approx_gain}
\objgain = \ap(o,W_i) - \ap(o,W_{i-1}) =  \frac{\contribution(\ar(o,qn))  - \contribution(\ar(o,qo))}{|W_i|}
\end{equation}

Here, if $o$ does not satisfy the query constraint $\constraint(q)$,
the contribution of $q$ to the popularity of $o$,
$\contribution(\ar(o,qn)) = 0 $.

Let $\maxobjgain$ denote the maximum popularity gain among all
objects (in the current window $W_i$).
Then the popularity of any object $o^{\prime} \in O \backslash
\res_{i-1}$ can be at most $\popmplusoneth + \maxobjgain$, where
$\popmplusoneth$ is the \mnth{(m+1)} highest approximate popularity
in the previous window $W_{i-1}$.
In other words, if the updated popularity of an object $o \in
\res_{i-1}$ is higher (better) than $\popmplusoneth + \maxobjgain$,
then such an $o$ is guaranteed to remain in $\res_i$, which inspires
the design of the object-level safe rank $\SR$ shown in the following
equation: 
\begin{equation}\label{eqn:sr}
\popqo + \frac{N-\SR+1}{|W_i|} \ge \popmplusoneth + \maxobjgain
\end{equation}

Since a lower rank indicates a higher contribution to the popularity,
the gain will be maximized when the difference between $\ac(\ar(o,qn))$ and
$\ac(\ar(o,qo))$ is maximized.
Therefore, the goal of minimizing the updates of objects in $R_{i-1}$
(set at the beginning of this section) can be reduced to the
challenge of how to compute a tight estimation of $\maxobjgain$.

A naive approach to estimate $\maxobjgain$ is to overestimate
$\ac(\ar(o,qn))$ as $N-1+1$ (the rank of $o$ is ``1'' for $qn$)
and underestimate $\ac(\ar(o,qo))$ as ``0'' ($o$ does not satisfy
$\constraint(qo)$).
The safe rank $\SR$ for $W_i$ can then be computed using the naive
maximum gain value in Eqn.~\ref{eqn:sr}.
However, such an estimation of the maximum gain is too loose, and may
not have any pruning capacity, especially if the queries $qn$ and
$qo$ are close to each other (an object that is ranked very
high for $qn$ but ranked very low for $qo$ may not exist).

\subsubsection{Block-level popularity gain}
Since the objects are arranged blockwise in an \irr index, and
each object $o$ is sorted by its lower bound rank $\LR(o,c_q)$ in
ascending order, we are motivated to define and utilize a 
\emph{block-level gain} as the first step in finding a tighter
estimation of the maximum object-level gain.


In particular, a block-level maximum gain $\maxblkgain$ is computed,
such that $\maxblkgain \ge \maxobjgain$, which can be used to find
the block-level safe rank, $\BSR$.
If the rank of any result object is not better than $\BSR$ for $qn$,
then an object-level maximum gain, $\maxobjgain$ is computed.
The object-level safe rank $\OSR$ can be computed using this value,
where $\OSR \ge \BSR$, as a lower value of rank implies a higher
gain.
If the rank of any result object is still not safe, then the
validation objects (proposed in Sec.~\ref{subsec:vo}) must be checked
to decide if $\res_{i-1}$ needs to be updated.
Here, some part of the safe rank calculations can be reused to find
the validation objects, which will be explained in
Section~\ref{subsec:vo}.

Next, we define the block-level gain and propose a technique to
compute the block-level maximum gain, from which a block-level safe
rank $\BSR$ can be computed in the following section.

\myparagraph{Block-level gain computation}
\begin{figure}
\centering
\includegraphics[height=.85in]{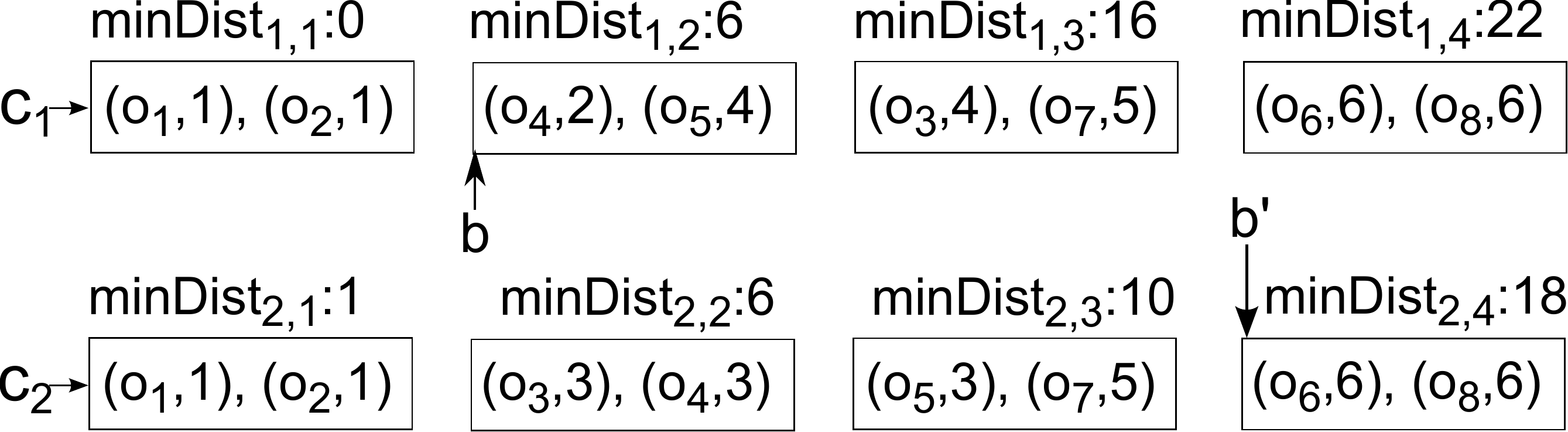}
\vspace{-6pt}
\caption{Upper bound rank computation of a block}
\label{fig:index}
\end{figure}

Given a block $b$ from the rank list of $qn$, the block-level gain
$\blkgain$ is an overestimation of the gain of the objects $o \in b$,
such that $\blkgain \ge \objgain$.
As the gain is maximized when the difference between $\ac(\ar(o,qn))$
and $\ac(\ar(o,qo))$ is maximized, a technique to compute $\blkgain$
can be actualized by finding: (i) a lower bound estimation of the rank of
any object $o \in b$ for $qn$, namely $\LBR(b,qn)$, where,
$\LBR(b,qn) \le \ar(o,qn)$; and (ii) an upper bound estimation of the
rank that any object $o \in b$ can have for $qo$, denoted as
$\UBR(b,qo)$, such that $\ar(o,qo) \le \UBR(b,qo)$.

Since the objects are sorted in ascending order of lower bound
ranks in the \irr index, the lower bound rank of the first entry of
$b$ is implicitly $\LBR(b,qn)$.
Here, $\forall o \in b, \LBR(b,qn) \le \LR(o,qn)$ holds by definition. 

Next, for the same block $b$ of the rank list of $qn$, finding the
maximum rank $\UBR(b,qo)$ that any object $o \in b$ can have for
$qo$ is needed.
To achieve this, a block $b^\prime$ is found such that all of the
objects $o \in b$ are guaranteed to be in the rank list of $qo$
before $b^\prime$.
As the objects are sorted by $\LR(o,c_{qo})$ in the rank list of
$c_{qo}$, $\LR(o^{\prime},c_{qo})$ is guaranteed to be greater
than that of any object in $b^\prime$, where $o^{\prime}$ is the
first entry of $b^\prime$.
Therefore, $\LR(o^{\prime},qo)$ is taken as the upper bound
estimation, $\UBR(b,qo)$.

For a tight estimation of $\UBR(b,qo)$, the block $b^\prime$ with the
smallest $\LR(o^{\prime},qo)$ must be found.
As the objects and blocks of a rank list are sorted by the minimum
Euclidean distance from the corresponding cell
(Section~\ref{sec:index}), and $\forall o \in b,$ $\mindist(o,c_{qo})
\le \maxdist(b,c_{qo})$, a binary search over the blocks of the rank
list of $qo$ is performed to find the first position of the block
$b^{\prime}$ where $\maxdist(b,c_{qo}) \le
\mindist(b^{\prime},c_{qo})$.
Here, $\maxdist(b,c_{qo})$ is computed as the maximum
Euclidean distance between the minimum bounding rectangle of the
objects $o \in b$ and cell $c_{qo}$.

\begin{example}
Computing the block gain is explained with the example in
Figure~\ref{fig:index}.
Let $c_1$ and $c_2$ be the cell where query $qn$ and $qo$ arrive
respectively.
Assume that the constraint of both queries are satisfied by all the
objects for ease of explanation.
Let $b = \langle (o_4,2), (o_5,4) \rangle$ be the block of the rank
list of $qn$ currently under consideration.
Here, $\LBR(b,qn) = 2$, which is the lower bound of the first entry
of $b$.
Let $\maxdist(b,c_2) = 14$, computed from the MBR of block $b$ and
cell $c_2$.
Now, as the objects and the blocks in the rank list of $c_2$ are sorted by their
minimum Euclidean distance from $c_2$, a binary search is performed with the value $14$ over the
$\mindist$ of the blocks in $c_2$.
Note that, $b$ is a block in the rank list of $c_1$, consisting of
the objects $o_4$ and $o_5$.
Here, we get $b^\prime = \langle (o_6,6), (o_8,6) \rangle$, as
$\mindist(b^\prime,c_2) = 18$, which is the smallest value of
$\mindist$ greater than $14$, shown with an arrow.
So, $\UBR(b,qo) = 6$ is the lower bound rank of the first entry of
$b^\prime$.
\end{example}

\subsubsection{Block-level safe rank}
\setlength{\algomargin}{1.2em}
\begin{algorithm}[t]
\caption{{\sc Block\_safe\_rank} }
\label{algo:bsr}
\begin{small}
{\bf Input:}\\ 
\hspace{2mm} $\popmplusoneth$ - \mnth{(m+1)} highest popularity of $W_{i-1}$,\\
\hspace{2mm} $\popqo$ - \mnth{m} highest popularity from $\res_{i-1}$ after updating \\
\hspace{2mm} for $qo$, and $PQ$ - a max-priority queue.\\
{\bf Output:}~~Block based safe rank - $\BSR$\\
$b \leftarrow$ first block in the rank list of $c_{qn}$.\\
$\maxblkgain \leftarrow 0$\\
\Do{$b$ cannot have a better gain than $\maxblkgain$}
{
$\LBR(b,qn) \leftarrow \LR(o,c_{qn})$ of the first entry $o$ from $b$.\\
$\maxdist(b,c_{qo}) \leftarrow$ Maximum Euclidean distance between $b$ and $c_{qo}$.\\
$b^{\prime} \leftarrow$ First position of the block of $c_{qo}$, where $\maxdist(b,c_{qo}) \le \mindist(b^{\prime},c_{qo})$.\\
$\UBR(b,qo) \leftarrow$ $\LR(o^{\prime},c_{qo})$ of the first entry
$o^{\prime}$ of $b^{\prime}$.\\
$\blkgain  \leftarrow \frac{\displaystyle \contribution(\LBR(b,qn)) - \contribution(\displaystyle \UBR(b,qo)) }{\displaystyle |W|} $\\
{\sc Enqueue} ($PQ,b,\blkgain)$\\
$\maxblkgain \leftarrow \gain_{top(PQ)}$\\
$b \leftarrow$ {\sc Next} ($c_{qn}$)\\
}
$\BSR$ $\leftarrow$ Compute from $\popmplusoneth$,$\popqo$, $\maxblkgain$ as Eqn.~\ref{eqn:block_safe_rank}.\\
{\sc return} $\BSR$
\end{small}
\end{algorithm}

By making use of the values $\LBR(b,qn)$ and $\UBR(b,qo)$ of block
$b$, a block-level estimation of the maximum gain for $W_i$ can found,
and a block-level safe rank $\BSR$ can be computed, as shown 
in Algorithm~\ref{algo:bsr}.
Algorithm~\ref{algo:bsr} shows the steps needed to compute the
block-level safe rank by finding the maximum gain of a block using
the rank lists of $qn$ and $qo$.
A max-priority queue $\var{PQ}$ is used to keep track of blocks that
must be visited, where the key is $\blkgain$.
Here, $\blkgain$ is an overestimation of the gain of the objects in
$b$.
For any object $o \in b$, $\objgain \le \blkgain$, is computed in
Line 3.13 as - 
$$
\blkgain  \leftarrow \frac{\displaystyle \contribution(\LBR(b,qn)) - \contribution(\displaystyle \UBR(b,qo)) }{\displaystyle |W_i|}
$$

Recall that in the \irr index, each object $o$ in the rank list is
sorted in ascending order of the lower bound rank w.r.t.
the cell $c_{qn}$, and the traversal starts from the beginning of the
rank list of $c_{qn}$ so that the objects with a higher gain are most
likely to be explored first.
The traversal continues until the subsequent blocks of the rank lists
of $qn$ cannot have a better gain than the current maximum gain
$\maxblkgain$ found so far.
Here, the {\bf terminating condition} of Line 3.17 is: $$
\frac{\contribution(\LBR(b,qn))}{|W_i|} < \maxblkgain
$$
 
Lastly, in Line 3.18 the maximum gain value $\maxblkgain$ is used to
compute the block-level safe rank as follows: 
\begin{equation}\label{eqn:block_safe_rank}
\ap(o_m,W_{i-1}\backslash qo) + \frac{N-\BSR+1}{|W_i|} \ge \popmplusoneth + \maxblkgain
\end{equation}

\subsubsection{Object-level safe rank}
\setlength{\algomargin}{1.2em} 
\begin{algorithm}[t]
\caption{{\sc Object\_safe\_rank} }
\label{algo:osr}
\begin{small}
{\bf Input:}\\
\hspace{2mm} $\popmplusoneth$ - \mnth{(m+1)} highest popularity of $W_{i-1}$,\\
\hspace{2mm} $\popqo$ - the updated \mnth{m} highest popularity from $\res_{i-1}$\\
\hspace{2mm} after removing $qo$, and $PQ$ - a max-priority queue from 
\hspace{2mm} {\sc Block\_safe\_rank}. \\
{\bf Output:}~~Object-level safe rank, $\OSR$\\
\While{$PQ$ not empty}
{
	$E \leftarrow$ {\sc Dequeue} ($PQ$)\\
	\uIf{E is object}
	{
		$\maxobjgain \leftarrow \gain_E$\\
		{\sc break}\\
	}
	\Else
	{
		\For{o in E}
		{
			$\objgain \leftarrow \frac{\displaystyle \contribution(\ar(o,qn)) - \contribution(\ar(o,qo)) }{\displaystyle |W|} $\\
			{\sc Enqueue} ($PQ,o,\objgain$)\\
		}
	}
}
$\OSR$ $\leftarrow$ Compute from $\popmth$,$\popqo$,$\maxobjgain$ (by \zf{Eqn.~\ref{eqn:sr}}).\\
{\sc return} $OSR$
\end{small}
\end{algorithm}

If the rank of any object $o \in \res_{i-1}$ for $qn$ is not smaller
(better) than the block-level safe rank $\BSR$, then the object-level
safe rank is computed, where $\OSR \ge \BSR$ is used to further
determine whether the result needs to updated or not (Lines 2.18 -
2.23 in Algorithm~\ref{algo:approx}).

Algorithm~\ref{algo:osr} shows a best-first approach to compute the
maximum object gain using the same priority queue $\var{PQ}$
maintained in the block-level computation.
In each iteration, the top element $E$ of $\var{PQ}$ is dequeued from
$\var{PQ}$.
If $E$ is a block, the approximate rank of each object $o \in E$ for
$qn$ and $qo$ is computed using the corresponding lower bound rank
in the rank lists.
The objects are then enqueued in $\var{PQ}$, and indexed by the gain
computed using Eqn.~\ref{eqn:approx_gain}.
If $E$ is an object, then the gain is returned as the maximum object
level gain $\maxobjgain$ (Lines 4.8 - 4.9).
The object-level safe rank, $\OSR$, is then computed in the same manner
as Eqn.~\ref{eqn:block_safe_rank} with the value $\maxobjgain$.

\subsection{Validation objects}\label{subsec:vo}
If the rank of any object $o \in \res_{i-1}$ is not safe, a
set of validation objects $\var{VO}$ is found such that, as long as
$\forall vo \in \var{VO}$, $\ap(o,W_i) \ge \ap(vo,W_i)$, $o$ is a
valid result object of $\res_i$.
We present an efficient approach to incrementally identify $\var{VO}$.
Furthermore, we show that if the result needs to be updated, the new
result objects also must come from $\var{VO}$.

First, after a new query $qn$ arrives, the approximate rank for
each object $o \in \res_{i-1}$ is computed, and the appropriate
popularity scores are updated.
Let the updated \mnth{m} highest approximate popularity from
$\res_{i-1}$ be $\popmth$ (Line 2.19 of Algorithm~\ref{algo:approx}).
The priority queue $\var{PQ}$ maintained for safe rank computation is
used to find the set $\var{VO}$ of validation objects, where
$\popmth$ is used as a threshold to terminate the search.

A best-first search is performed using $\var{PQ}$ to find the objects
that have gain high enough to be a result.
Specifically, if the dequeued element $E$ from $\var{PQ}$ is a block,
the $\ar$ of each object $o$ in $E$ is computed for $qn$ and $qo$ in
the same manner as described for the object-level safe rank
computation.
As the popularity of an object $o \in O \backslash \res_{i-1}$ can be
at most $\popmplusoneth + \objgain$, an object $o$ is included in the
validation set if $o$ satisfies the following condition:
\begin{equation}\label{eqn:vo}
\popmplusoneth + \objgain \ge \popmth
\end{equation}

As $\var{PQ}$ is a max-priority queue which is maintained for the
gain of the objects and the blocks, the process can be safely
terminated when the gain of a dequeued element $E$ does not satisfy
the condition in Equation~\ref{eqn:vo}.

If no validation object is found, this implies that there is no object
that can have a higher popularity than the current results.
In this case, the result set $\res_{i-1}$ (of previous window
$W_{i-1}$) remains unchanged, and is the result of current window
$W_i$.
Otherwise, the popularity of each object in $\res_{i-1} \backslash
\res_i$ needs to be checked against the popularity of the validation
objects $vo \in \var{VO}$ to update the result.

\subsubsection{Updating results}
As described in Section~\ref{subsec:vo}, the set of validation
objects $\var{VO}$ is computed such that no object $o \backslash
\var{VO}$ can have a higher popularity than any of the objects in
$\res_{i-1}$.
Therefore, only objects in $\var{VO}$ are considered when
updating the result set.
To update the results using the objects $vo \in \var{VO}$,
the popularity of $vo$ for the current window must be computed.
Therefore, an efficient technique to compute the popularity of the 
validation objects is now presented.

\myparagraph{Computing $\ap$ of the validation objects}
\label{subsec:pop_compute}
As the popularity gain of each $vo \in \var{VO}$ has already been
computed as described in Section~\ref{subsec:vo}, it is sufficient to
find the $\ap(vo,W_{i-1} \backslash qo)$ and use it to compute
$\ap(vo,W_i)$.
Since the popularity of every object for every window is not
computed, a straightforward way to compute $\ap(vo,W_{i-1} \backslash
qo)$ is to find the rank of $vo$ for each $q \in W_{i-1} \backslash
qo$ using the corresponding rank lists.
However, this approach is computationally expensive, especially when
the window size is large.
Moreover, if $vo$ was a validation object or a result object in a
prior window $W_{i-y}$, then the same computations are repeated
unnecessarily for the queries shared by the windows (the queries
contained in $W_i \cap W_{i-y}$).

Therefore, if $\ap$ of a result or a validation object is
computed for a window $W_{i-y}$, the aim is to reuse this computation
for later windows in an efficient way.
This can be accomplished by storing the popularity of a subset of
``necessary'' objects from prior windows for later reuse.
We show that the choice of these limited number of windows is
optimal, and storing the popularity for any additional windows cannot
reduce the computational cost any further.

\myparagraph{Choosing the limited number of prior windows}
The popularity computations can be reused if the number of
shared queries among the windows is greater than the number of
queries that differ. 
Otherwise, the popularity must be computed for the window $W_i$ from
scratch rather than reusing the popularity computations from
$W_{i-y}$.
Specifically, let $Y$ be the number of shared queries among windows
$W_i$, $W_{i-y}$ ($Y = |W_i \cap W_{i-y}|$), $Q_o = W_{i-y}
\backslash W_i$, and $Q_n = W_i \backslash W_{i-y}$.
So in a count based window, $|Q_n| = |W_i| - Y$ and $|Q_o| = |W_i| -
Y$, as each time the sliding window shifts, a new query is inserted
and the least recent query is removed from the window.
If the number of computations required for the shared queries is
greater than the number of computations for $|Q_n| + |Q_o|$, i.e., $Y
\ge 2(|W_i|-Y)$, then computations can be reused.
So the number of shared queries, $Y$, should be greater than or equal
to $2|W_i|/3$ for efficient reuse.

\myparagraph{Reusing popularity computations}
If the condition $Y \ge 2|W_i|/3$ holds, the popularity of an object
$o$ computed for $W_{i-y}$ can be used as $\ap(o,W_i)$ as follows: 
\begin{equation}\label{eqn:reuse}
\begin{aligned}
\ap(o,W_i) = \ap(o,W_{i-y}) + \\
\frac{\sum_{qn \in Q_n} \contribution(\ar(o,qn)) - \sum_{qo \in Q_o} \contribution(\ar(o,qo))}{|W_i|}
\end{aligned}
\end{equation}

\myparagraph{Popularity lookup table}
A popularity lookup table is maintained with the popularity of the
result and validation objects for the most recent $2|W_i|/3$ windows.
If a validation object $vo$ of the current window $W_i$ is found in
the lookup table, the popularity is computed using Equation~\ref{eqn:reuse}.
Otherwise, the popularity of $vo$ is computed from the rank lists of
the queries in $W_i$.
The popularity $vo$ for $W_i$ is then added to the popularity lookup
table for later windows.

\myparagraph{Obtaining Results}
The objects $vo \in \var{VO}$ are considered one by one to update the
results.
After computing the popularity of an object $vo \in \var{VO}$, if
$\ap(vo,W_i) > \popmth$, then $vo$ is added to $\res_i$.
The set $\res_i$ is adjusted such that it contains $m$ objects with the highest $\ap$, and the value of $\popmth$
is adjusted accordingly. 
In this process, if the overestimated popularity of an object $vo$
computed with Eqn~\ref{eqn:vo} is less than the updated $\popmth$, 
that object can be safely discarded from consideration without
computing its popularity.

\subsection{Approximation error bound}\label{sec:error}
In this section we present the bound for approximation error of our
proposed approach.
Specifically we show that, for any object $o \in O$, and any window
$W$ of queries, the ratio between $\ap(o,W)$ and $\rho(o,W)$ is
bounded. 

\begin{lemma}\label{lemma:pop_error}
For any object $o \in O$, and a window $W$ of queries, the
approximation ratio is bounded by $1-\epsilon/2N$.\\ (i)
$\ap(o,W)/\rho(o,W)
\le 1-\epsilon/2N$ when $\rho(o,W) \ge \ap(o,W)$; and \\ (ii)
$\rho(o,W)/\ap(o,W) \le 1-\epsilon/2N$ when $ \ap(o,W) \ge \rho(o,W)$
always holds.
\end{lemma}

\begin{proof}
See Appendix~\ref{appendix:proof} 
\end{proof}

\section{Experimental Evaluation}\label{sec:exp}
In this section, we present the experimental evaluation for our
proposed approach to monitor the top-$m$ popular objects in a sliding
window of streaming queries.
As there is no prior work that directly answers this problem
(Section~\ref{sec:related}), we compare our approximate solution
(proposed in Section~\ref{sec:approx_method}), denoted by {\AP}, with
the baseline exact approach (proposed in Section~\ref{sec:baseline}),
denoted by {\BS}.

\subsection{Experiment Settings}
\myparagraph{Datasets and query generation}\label{subsec:dataset}
All experiments were conducted using two real datasets, (i) {\Auss}
dataset at a city scale and (ii)
{\Foursq}\footnote{\url{https://sites.google.com/site/yangdingqi/home/foursquare-dataset}}
dataset at a country scale.

The {\Auss} dataset contains $52,913$ real estate properties sold in
a major Metropolitan city in Australia between 2013 to 2015, collected
from the online real estate advertising
site\footnote{\url{http://www.realestate.com.au}}.
As a property can be sold multiple times over the period,
only the first sale was retained in the dataset.
The locations of the queries in the {\Auss} dataset were created
by using locations of $987$ facilities (train stations, schools,
hospitals, supermarkets, and shopping centers) in this region.
We generated two sets of queries from these locations, each of size
$20K$. Repeating queries were created using two different approaches:
(i) uniform; and (ii) skewed distribution respectively.
We denote the uniform and the skewed query set as \method{U} and 
\method{S}, respectively.
The radius of the queries are varied as an experimental parameter,
and is discussed further in Section~\ref{subsec:evaluation}.

The {\Foursq} dataset contains $304,133$ points of interest (POI) from
Foursquare\footnote{\url{https://foursquare.com}} in $34$ cities
spread throughout the USA.
The queries for the {\Foursq} dataset were generated using the user
check-ins.
From the check-ins of each user, we generated a query, where the
query location was the centroid of all the check-ins of that user,
and the query radius was set as the minimum distance that covers
these check-ins.
If a user has only one check-in record, we set the query location as
the check-in location, and the radius of the query is randomly
assigned from another user.
As a result, a total of $22,442$ queries were generated for the
{\Foursq} dataset.

 \begin{figure}
\centering
\subfloat[Aus]{\includegraphics[trim = 60mm 10mm 65mm 10mm, clip, height=1.56in]{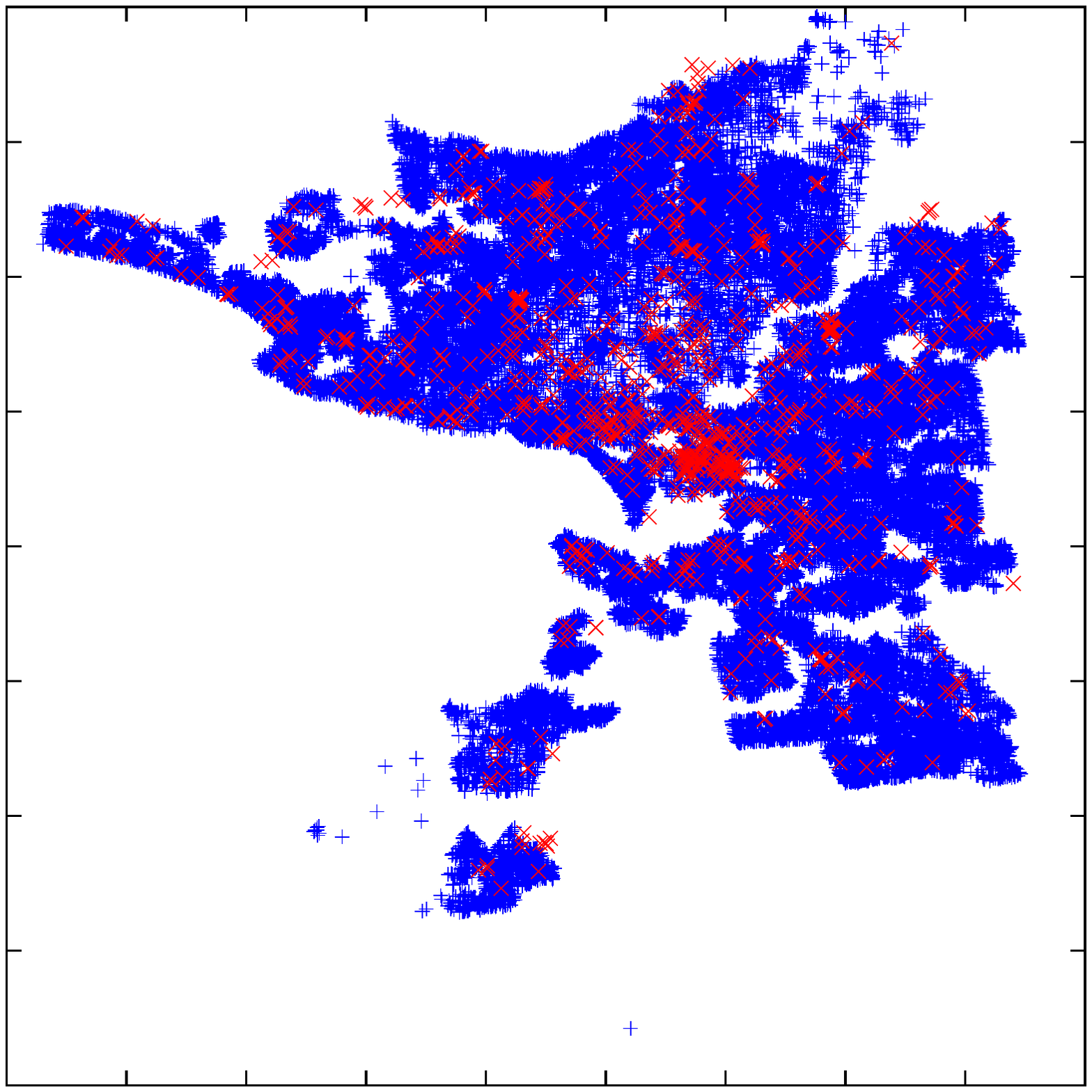}\label{fig:dataset1}}\hfill
\subfloat[Foursq]{\includegraphics[trim = 40mm 10mm 40mm 10mm, clip,height=1.6in]{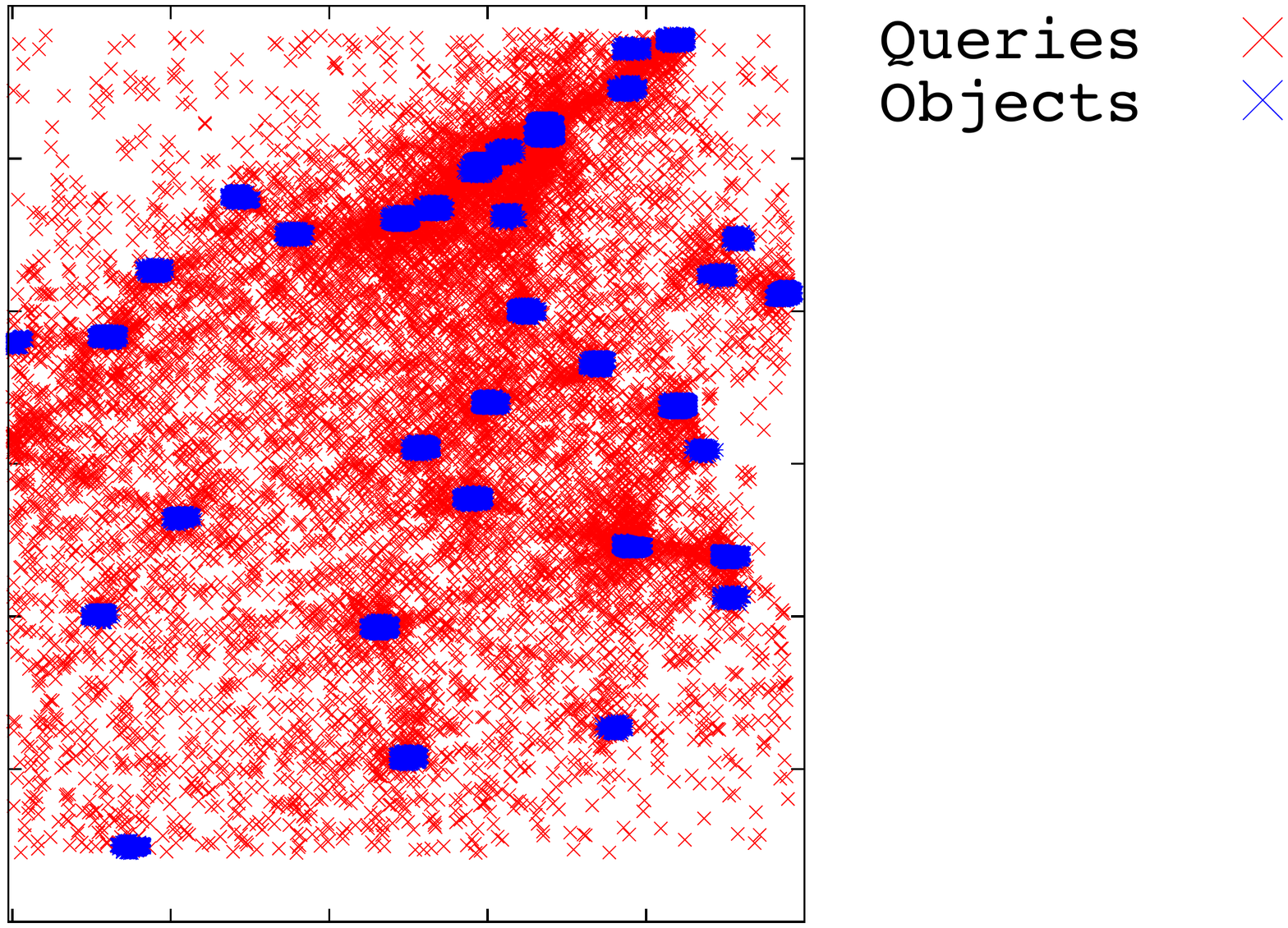}\label{fig:dataset2}}
\vspace{-10pt}
\caption{Dataset and query locations}
\vspace{-6pt}
\label{fig:dataset}
\end{figure}

Since {\Auss} dataset represents a real city-level data and queries are real facilities in that city, it is the best candidate for effectiveness study; as a result we conducted efficiency and effectiveness study on {\Auss}. Since the queries generated for {\Foursq} spread over the whole country, we find that it is more suitable for the efficiency and scalability study; nonetheless, we conducted the effectiveness study for {\Foursq} and most experiment results are shown in Appendix~\ref{appendix:exp}.

After initializing the sliding window, we evaluated the performance
of both {\BS} and {\AP} approaches for $10K$ query arrivals in the
stream, i.e., $10K$ shifts of the sliding window.
We repeated the process $50$ times, and report the mean performance.
For the {\Auss} dataset, the arrival order of the queries was
randomly generated.
For the {\Foursq} dataset, the arrival order of a query in the stream was
obtained from the most recent check-in time of the corresponding
user.
Figure~\ref{fig:dataset} shows the location distribution of the
objects, and the queries for both datasets, where the blue and the red
points represent object locations and query locations respectively.
Note that, for the {\Foursq} dataset, the POIs are clustered in large
cities (i.e., blue clusters).
As a user may check-in in different cities, the queries (which are
the centroid of the check-in locations) are distributed in different
locations across the US.

\begin{table}[h]
\caption{Parameters}
\vspace{-8pt}
\centering
\begin{tabular}{ll}
\toprule
{\bfseries Parameter} & {\bfseries Range} \\
\midrule
$W$ & $100,200,{\bf 400},800,1600$\\
$m$ & $1,5,{\bf 10},20,50,100$\\
Query radius (\%) & $1,2,{\bf 4},8,16$\\
$\epsilon$ & $1,2,{\bf 3},4,5$\\
$B$ & $32,64,{\bf 128},256,512$\\
\bottomrule
\end{tabular}
\label{table:param}
\end{table}

\myparagraph{Evaluation Metrics \& Parameter}
We studied the efficiency, scalability and effectiveness for both the
baseline approach ({\BS}), and the approximate approach ({\AP}) by
varying several parameters.
The parameter of interest and their ranges are listed in
Table~\ref{table:param}, where the values in bold represent the
default values.
For all experiments, a single parameter varied while keeping the rest
as the default settings.
For efficiency and scalability, we studied the impact of each
parameter on: the number of objects whose popularity are computed per
query (OPQ), to update the answer of \cmpo; and the runtime per query
(RPQ).

In order to measure the \emph{effectiveness} of our approximate
approach, the impact of each parameter on the following two metrics
are studied:
\begin{enumerate}[leftmargin=1em]
\itemsep 0em
	\item \textbf{Approximation ratio}: For a window $W$, for
	each $o_i \in \res$, $o^{\prime}_i \in \hat{\res}$, where $i$
	is the corresponding position of the object in the top-$m$
	results, we compute the approximation ratio as -
	$$ ratio = max \left( \frac{
	\hat{\rho}(o^{\prime}_i,W)}{\rho(o_i,W)} ,
	\frac{\rho(o_i,W)}{ \hat{\rho}(o^{\prime}_i,W)} \right) $$
	\noi We report the average approximation ratio of the sliding
	window by varying different parameters.
	As the approximate popularity of an object is an aggregation
	over the estimated ranks, the approximation ratio may not be
	``1'' (the best approximation ratio) even if the approximate
	result object list $\hat{\res}$ is exactly the same as that
	result list returned by the baseline.
	Therefore, we present the following metric to demonstrate the
	similarity of the approximate result object lists with the
	baseline.

	\item \textbf{Percentage of result overlap}: For a window
	$W$, let $|\res| = |O|$, where $\res$ is the sorted list of
	all of the objects according to their exact popularity.
	We report the similarity between the result list returned by
	the approximate approach, $\hat{\res}$ with $\res$ at
	different depths.
	Specifically, for each result object $o^{\prime}_i \in
	\hat{\res}$, where $|\hat{\res}| = m$, we record the
	percentage of objects in $\hat{\res}$, overlapping with the
	top-$k$ objects of $\res$, where $k$ is varied from $10$ to
	$200$.
	For instance, when $m =50$, we compute how many objects in
	the top-$50$ approximate result that also appear in the
	top-$50$, top-$75$, $\dots$, top-$150$ exact results.
	We report the percentage of the shared objects for different
	choices of $k$, averaged by $10,000$ shifts of the sliding
	window.
\end{enumerate}

\myparagraph{Setup}
All indexes and algorithms were implemented in C++.
The experiments were ran on a $24$ core Intel Xeon $E5-2630$ running
at $2.3$~GHz using {\gb{256}} of RAM, and $1$TB $6$G SAS $7.2$K rpm
SFF ($2.5$-inch) SC Midline disk drives.
All index structures are memory resident.



\subsection{Efficiency \& Scalability Evaluation}\label{subsec:evaluation}

\begin{figure}[t]
\centering
\subfloat[Objects computed]{\includegraphics[angle =-90,trim = 20mm 60mm 20mm 50mm, clip,width=0.25\textwidth]{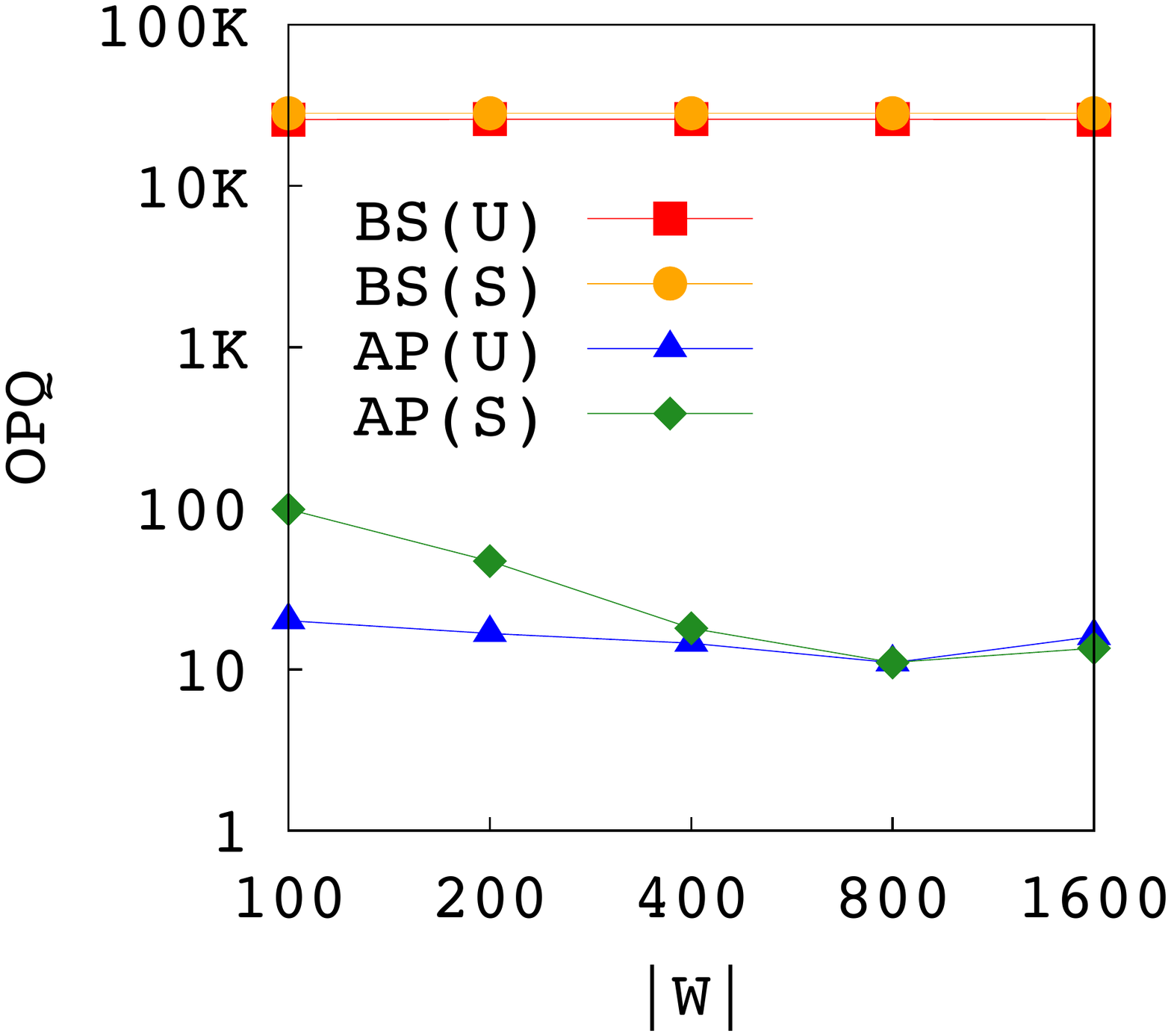}\label{fig:ow}}
\subfloat[Runtime]{\includegraphics[angle =-90,trim = 20mm 60mm 20mm 50mm, clip,width=0.25\textwidth]{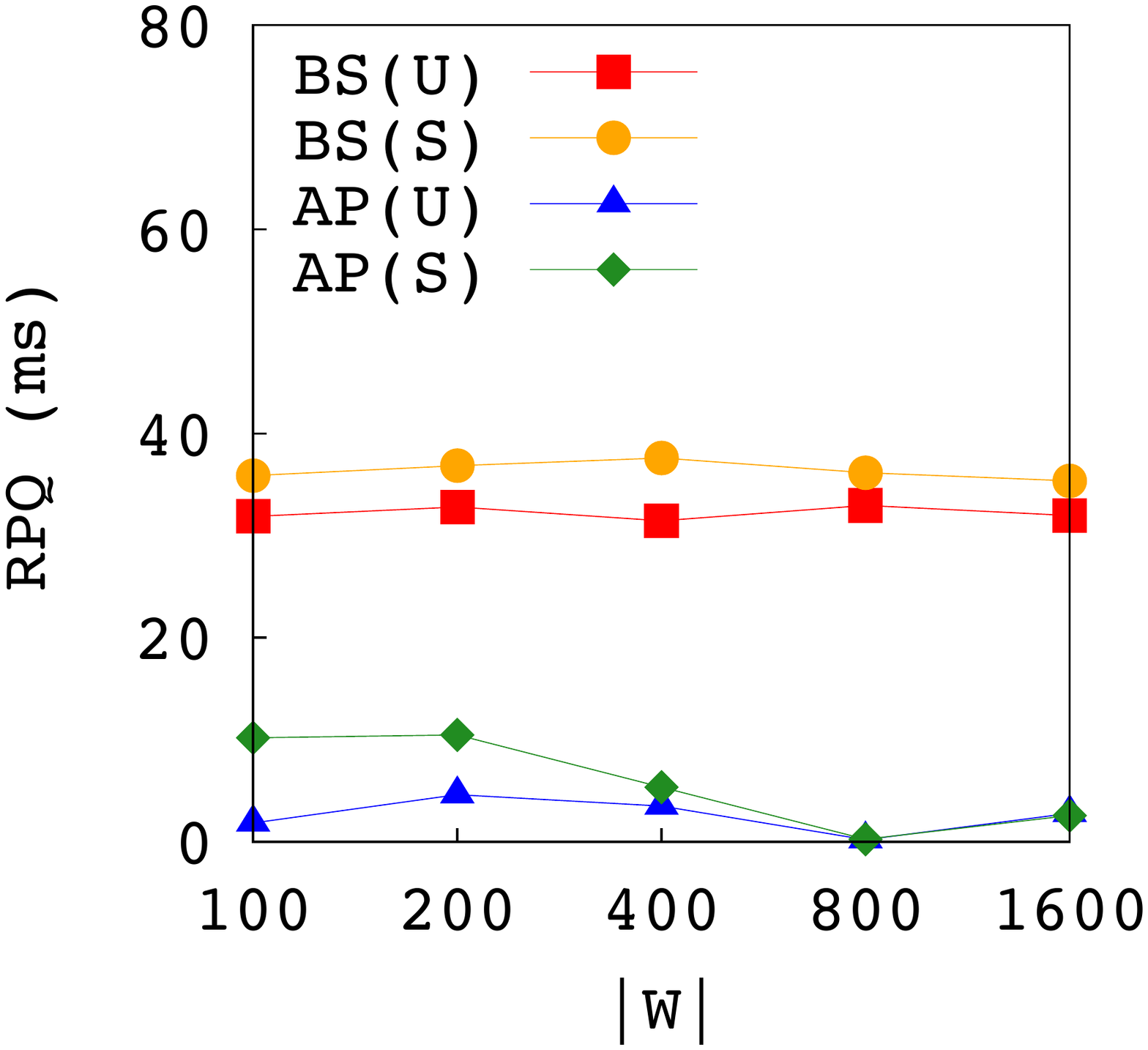}\label{fig:tw}}
\vspace{-8pt}
\caption{Effect of varying $|W|$ on \Auss dataset}
\vspace{-8pt}
\label{fig:w}
\end{figure}

\begin{figure}[t]
\centering
\subfloat[Runtime]{\includegraphics[angle =-90,trim = 20mm 60mm 20mm 50mm, clip,width=0.25\textwidth]{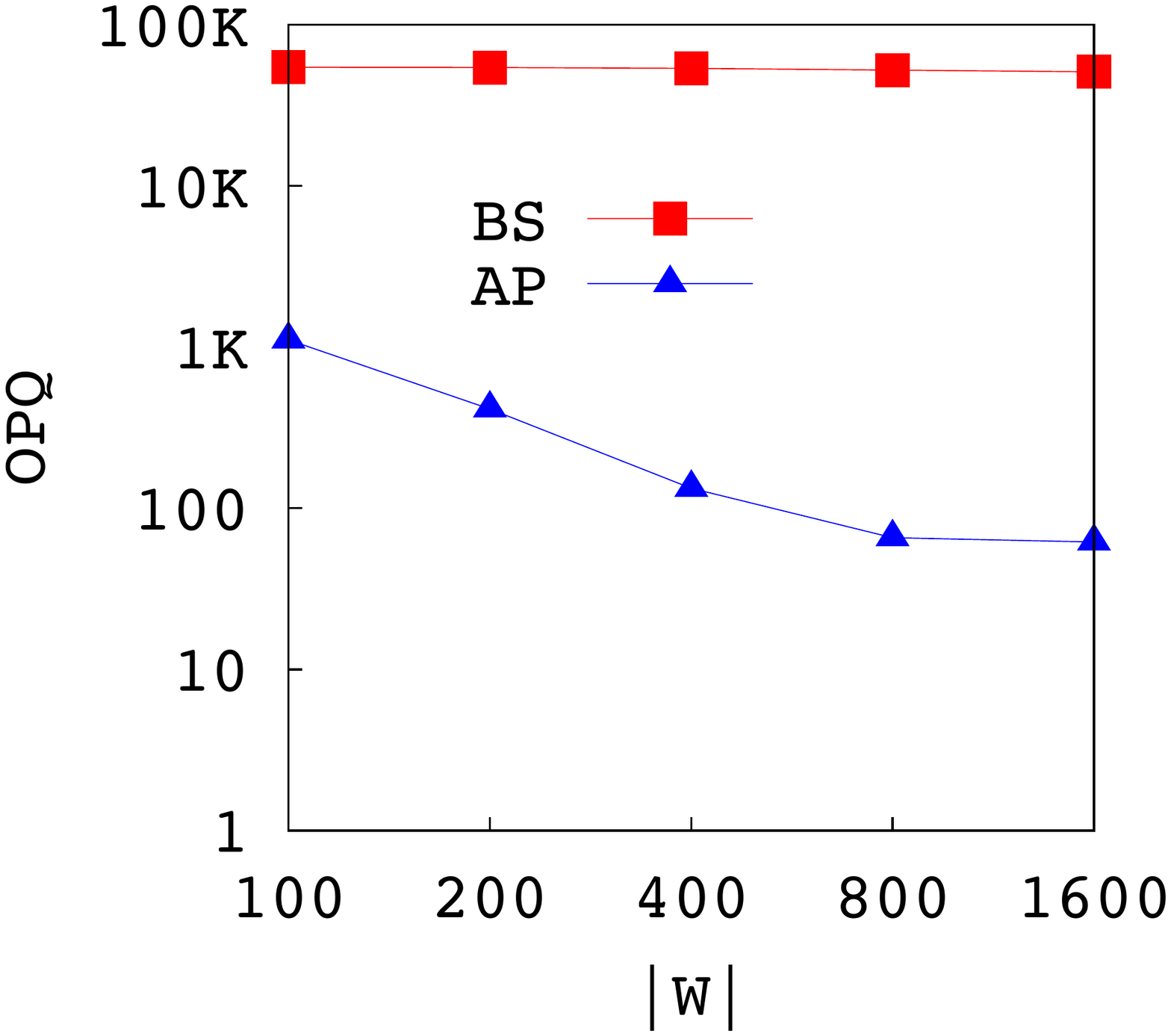}\label{fig:ow_foursq}}
\subfloat[Objects computed]{\includegraphics[angle =-90,trim = 20mm 60mm 20mm 50mm, clip,width=0.25\textwidth]{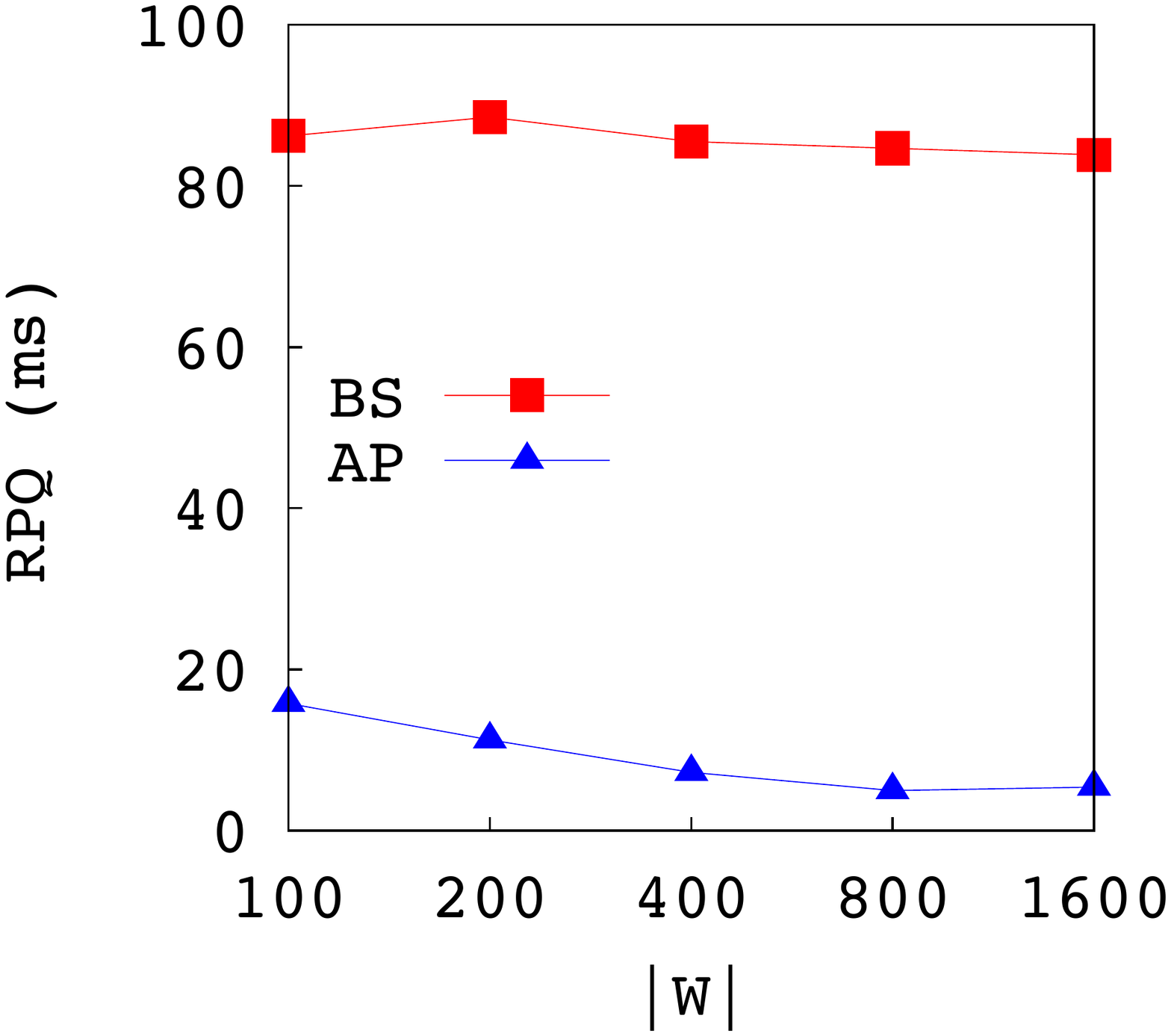}\label{fig:tw_foursq}}
\vspace{-8pt}
\caption{Effect of varying $|W|$ on \Foursq dataset}
\vspace{-8pt}
\label{fig:w_foursq}
\end{figure}

\begin{figure}[t]
\centering
\subfloat[Objects computed]{\includegraphics[angle =-90,trim = 20mm 60mm 20mm 50mm, clip,width=0.25\textwidth]{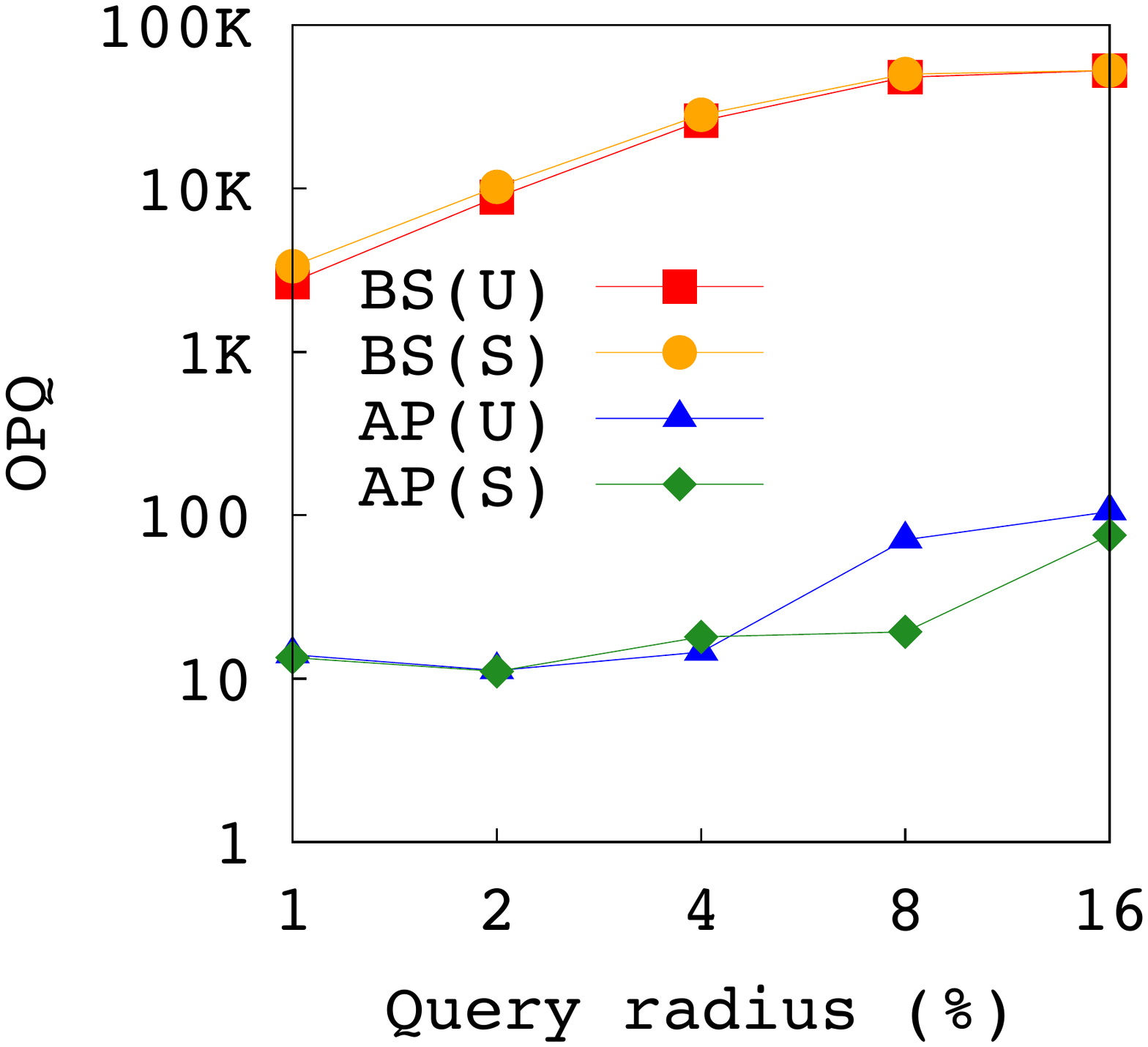}\label{fig:or}}
\subfloat[Runtime]{\includegraphics[angle =-90,trim = 20mm 60mm 20mm 50mm, clip,width=0.25\textwidth]{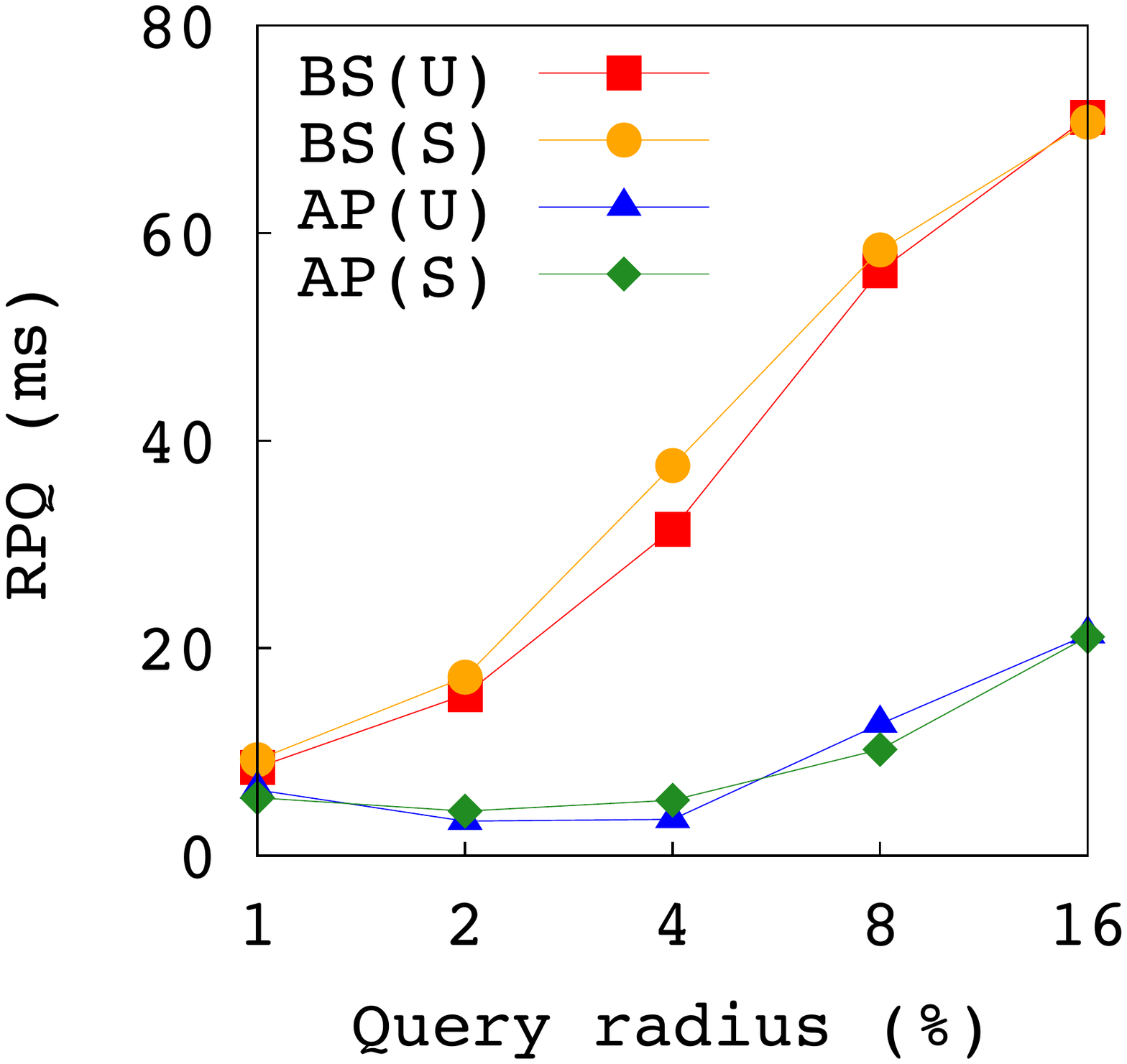}\label{fig:tr}}
\vspace{-8pt}
\caption{Effect of varying query radius on \Auss dataset}
\vspace{-8pt}
\label{fig:r}
\end{figure}

\begin{figure}[t]
\centering
\subfloat[Objects computed]{\includegraphics[angle =-90,trim = 20mm 48mm 20mm 55mm, clip,width=0.25\textwidth]{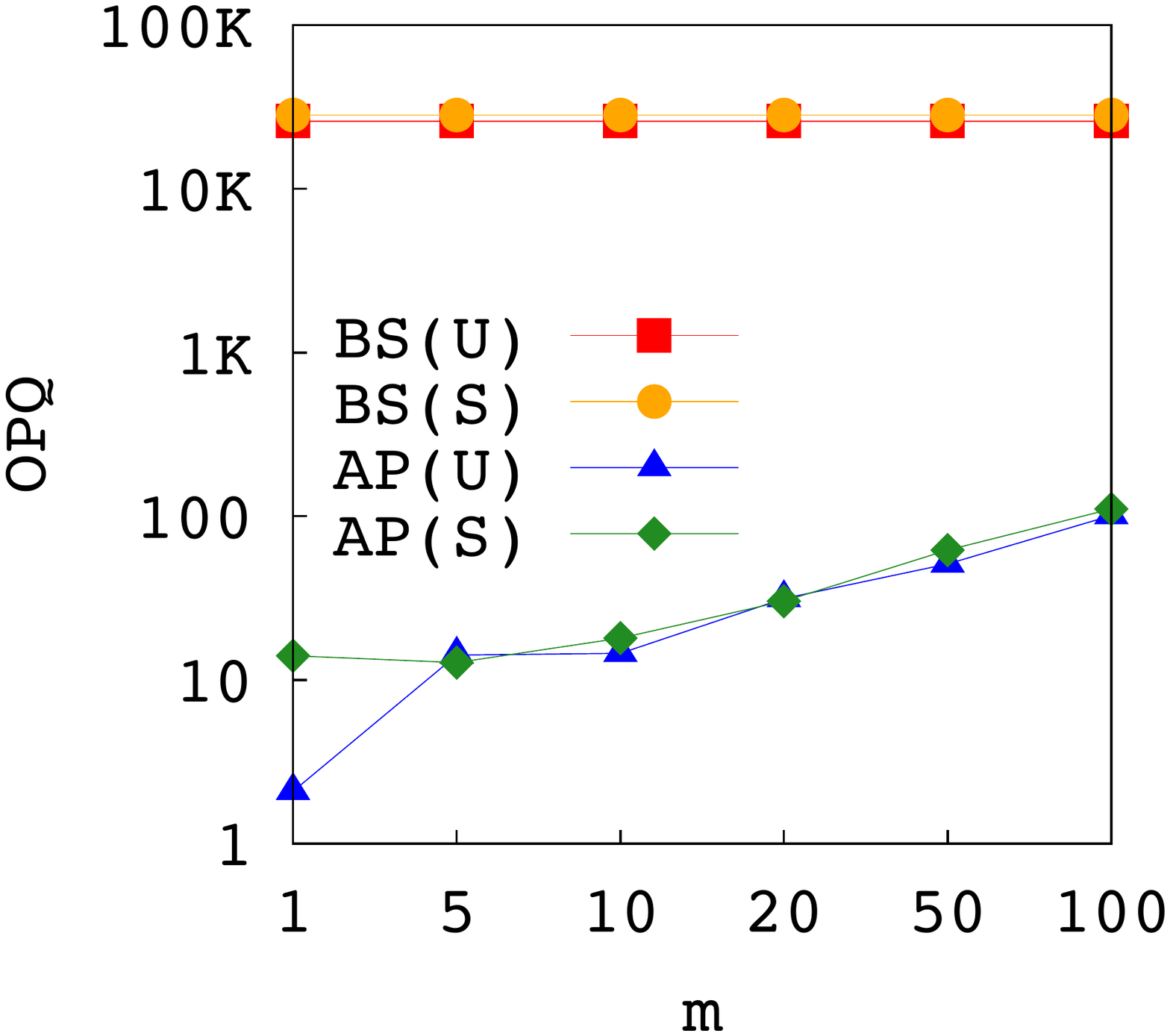}\label{fig:om}}
\subfloat[Runtime]{\includegraphics[angle =-90,trim = 20mm 48mm 20mm 55mm, clip,width=0.25\textwidth]{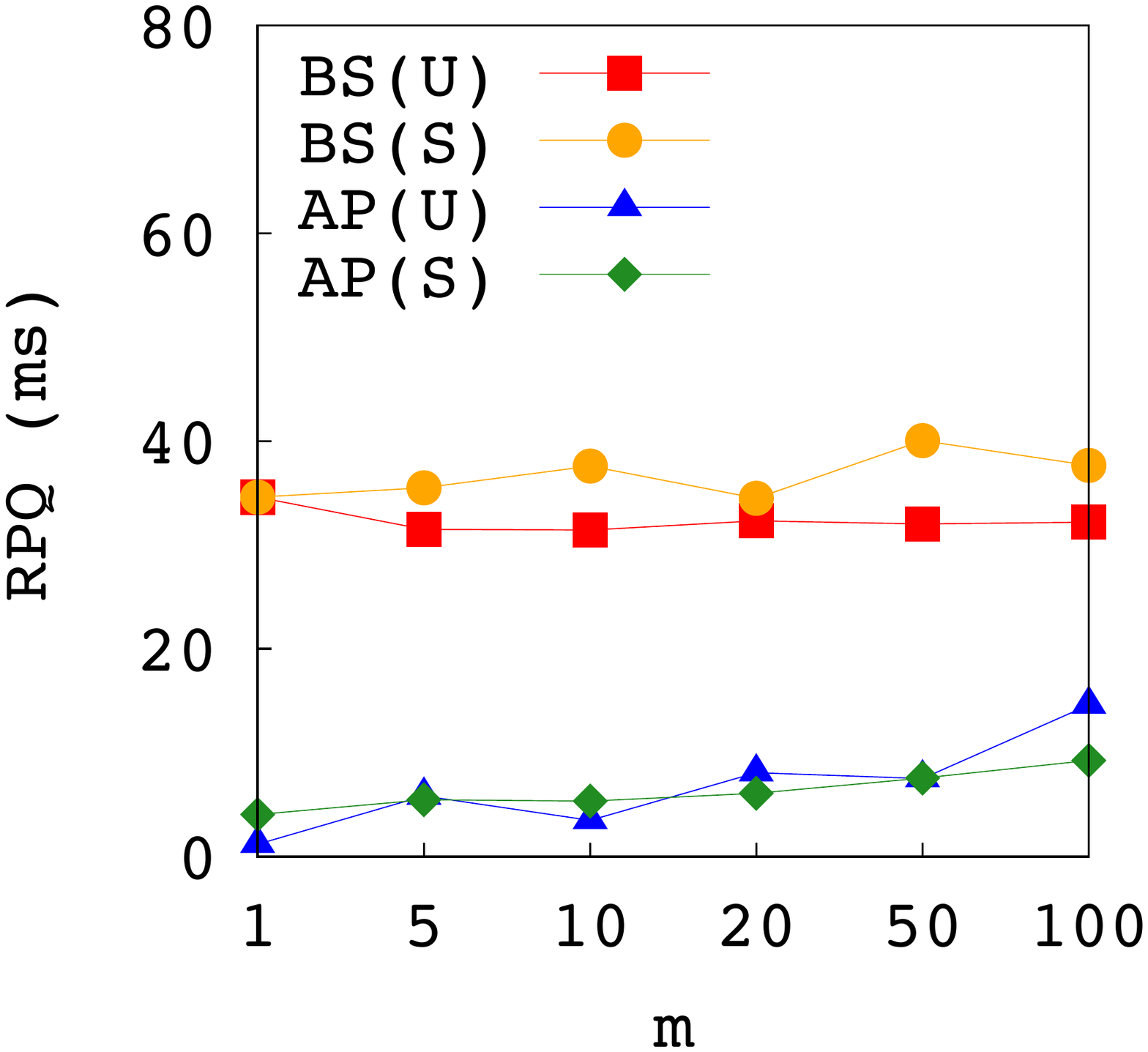}\label{fig:tm}}
\vspace{-8pt}
\caption{Effect of varying $m$ on \Auss dataset}
\vspace{-8pt}
\label{fig:m}
\end{figure}

\begin{figure}[t]
\centering
\subfloat[Objects computed]{\includegraphics[angle =-90,trim = 20mm 48mm 20mm 55mm, clip,width=0.25\textwidth]{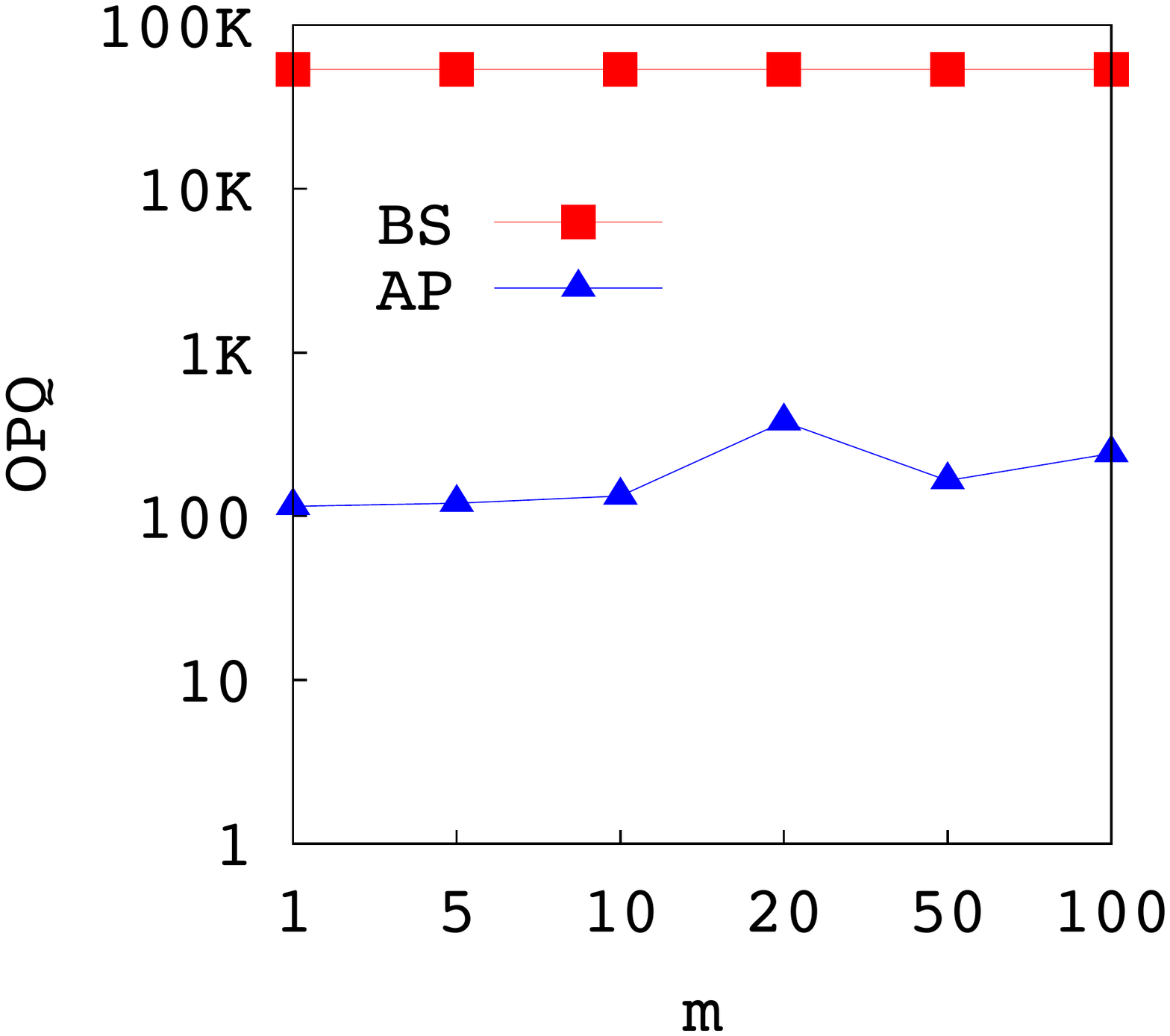}\label{fig:om_foursq}}
\subfloat[Runtime]{\includegraphics[angle =-90,trim = 20mm 48mm 20mm 55mm, clip,width=0.25\textwidth]{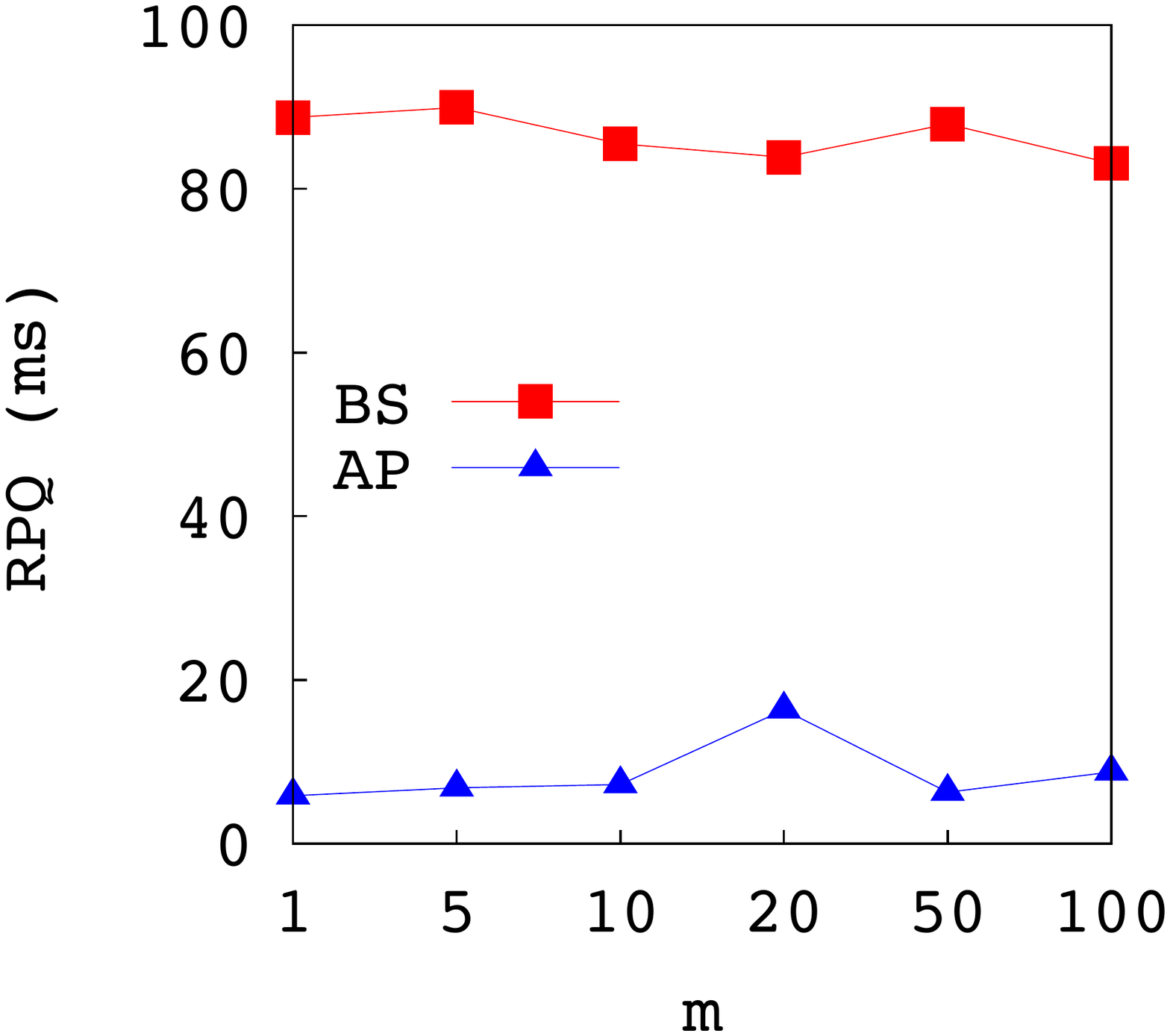}\label{fig:tm_foursq}}
\vspace{-8pt}
\caption{Effect of varying $m$ on \Foursq dataset}
\vspace{-8pt}
\label{fig:m_foursq}
\end{figure}

\begin{figure}[t]
\centering
\subfloat[Objects computed]{\includegraphics[angle =-90,trim = 20mm 48mm 20mm 55mm, clip,width=0.25\textwidth]{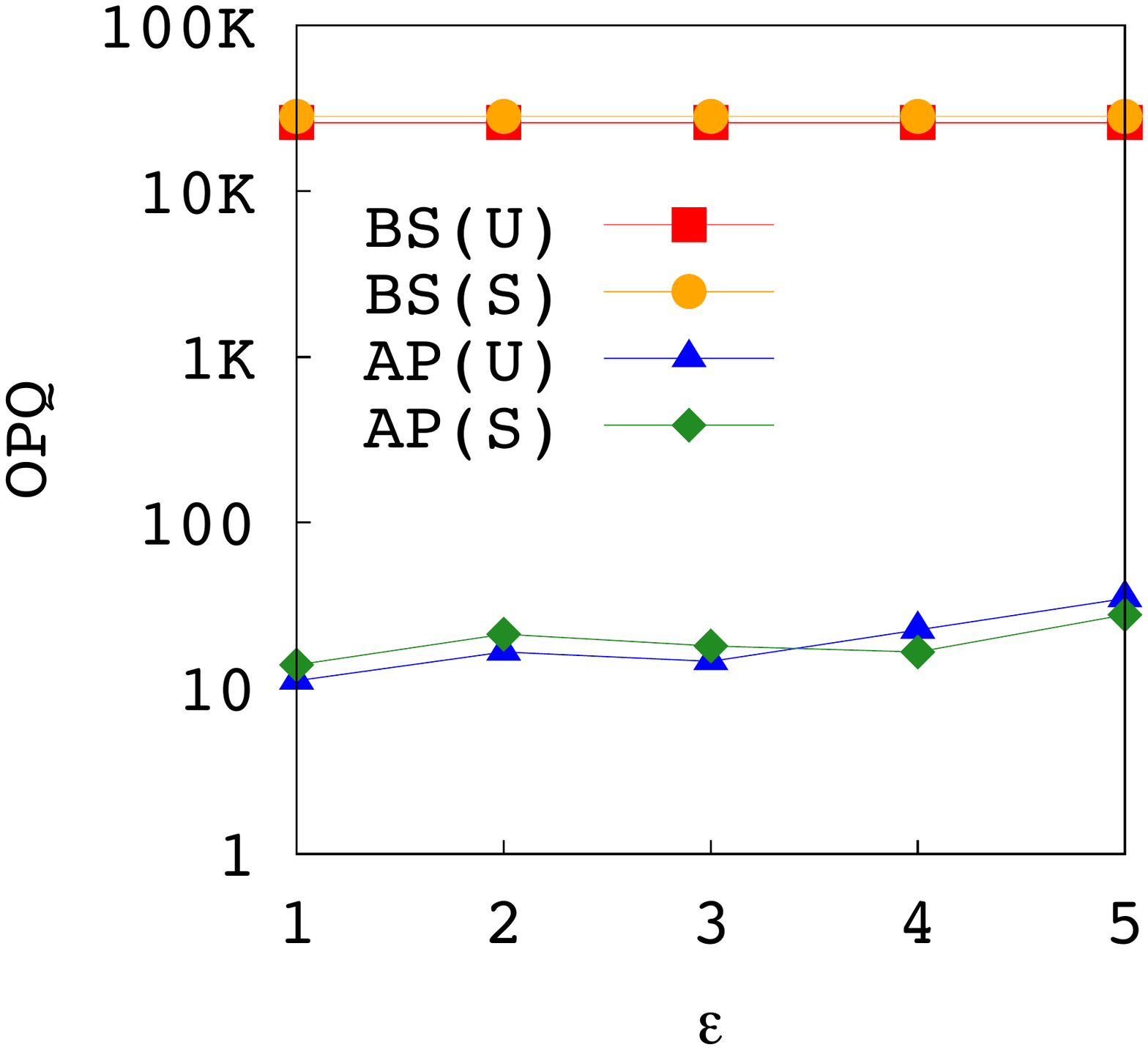}\label{fig:oep}}
\subfloat[Runtime]{\includegraphics[angle =-90,trim = 20mm 48mm 20mm 55mm, clip,width=0.25\textwidth]{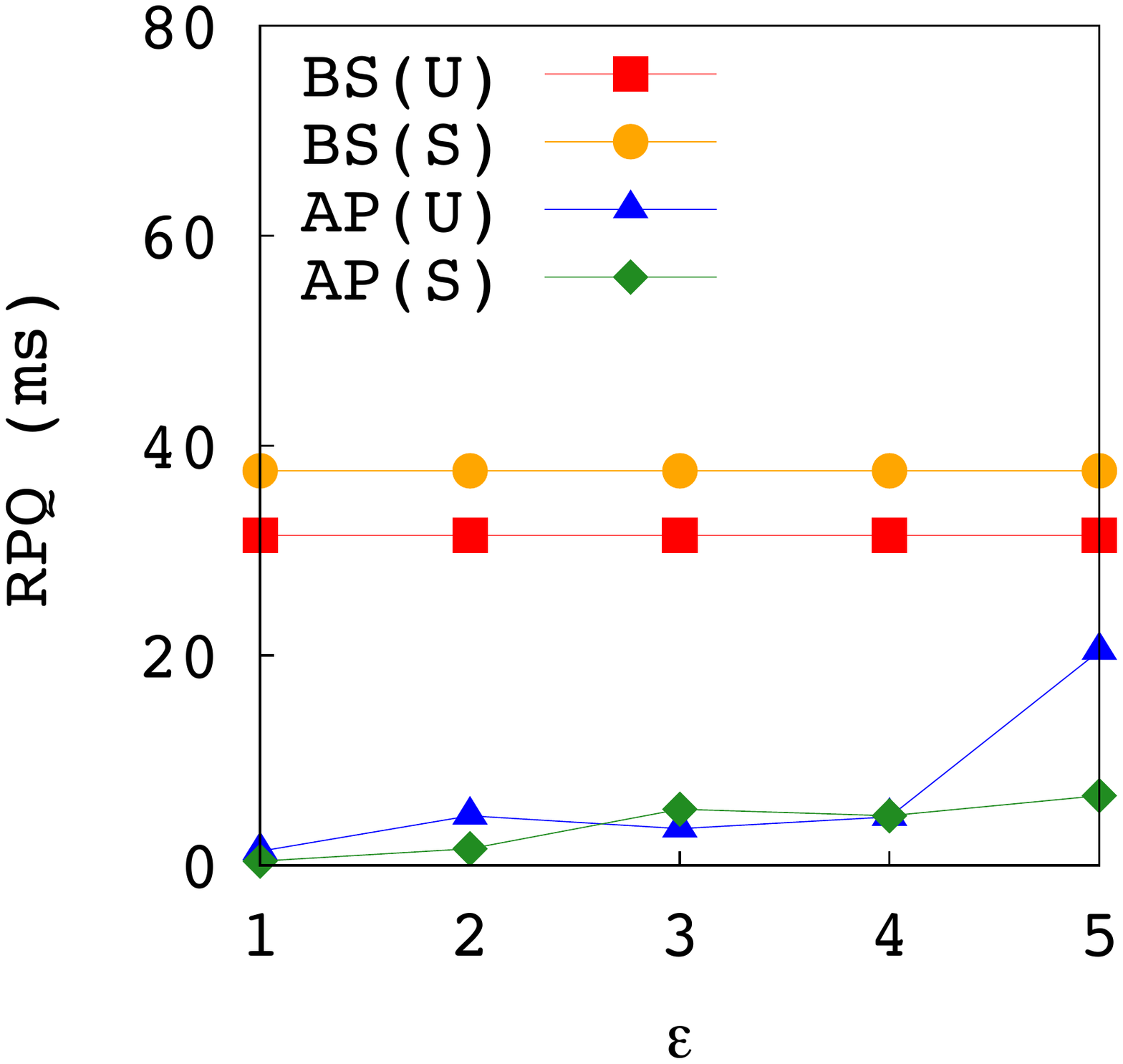}\label{fig:tep}}
\vspace{-8pt}
\caption{Effect of varying $\epsilon$ on \Auss dataset}
\vspace{-8pt}
\label{fig:epsilon}
\end{figure}

\begin{figure}[t]
\centering
\includegraphics[angle =-90,trim = 20mm 28mm 20mm 20mm, clip,width=0.32\textwidth]{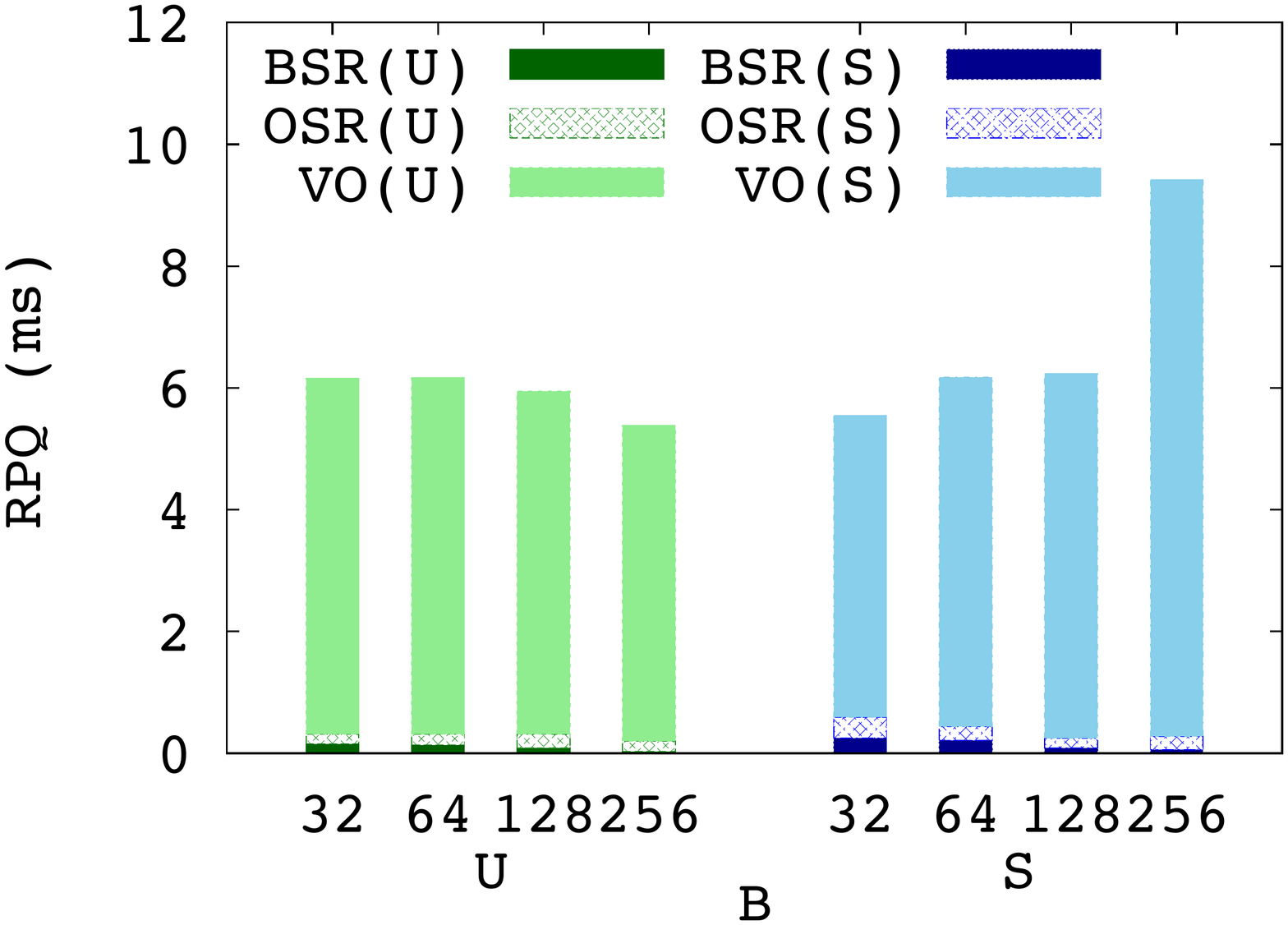}\label{fig:tb}
\vspace{-8pt}
\caption{Effect of varying $B$ on \Auss dataset}
\vspace{-8pt}
\label{fig:b}
\end{figure}

\begin{figure}[t]
\centering
\includegraphics[angle =-90,trim = 20mm 28mm 20mm 20mm, clip,width=0.32\textwidth]{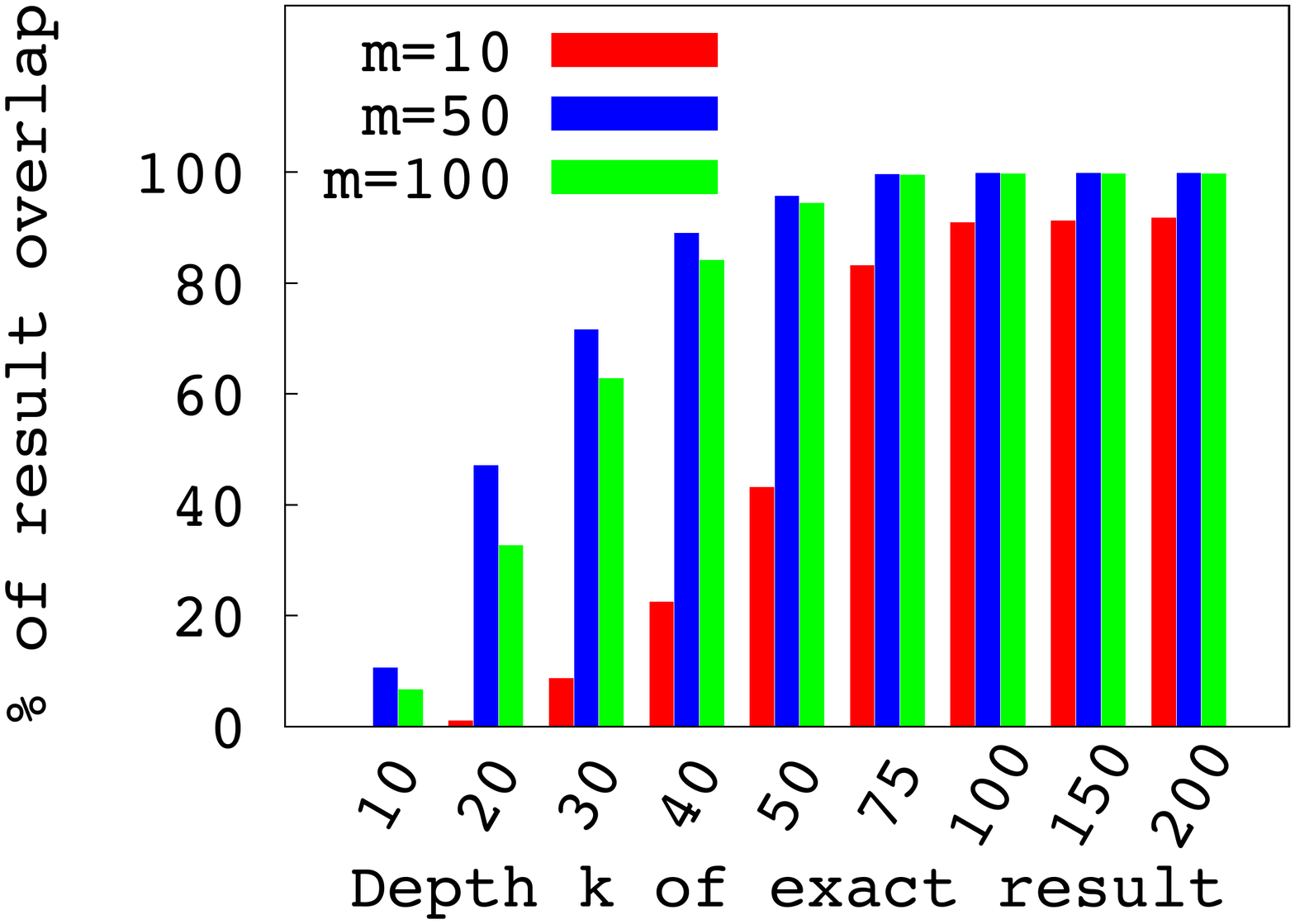}
\vspace{-8pt}
\caption{\% of result overlap for varying $m$ in \Auss dataset}
\vspace{-8pt}
\label{fig:overlap}
\end{figure}

\begin{figure}[!b]
\vspace{-12pt}
\centering
\includegraphics[angle =-90,trim = 20mm 28mm 20mm 20mm, clip,width=0.32\textwidth]{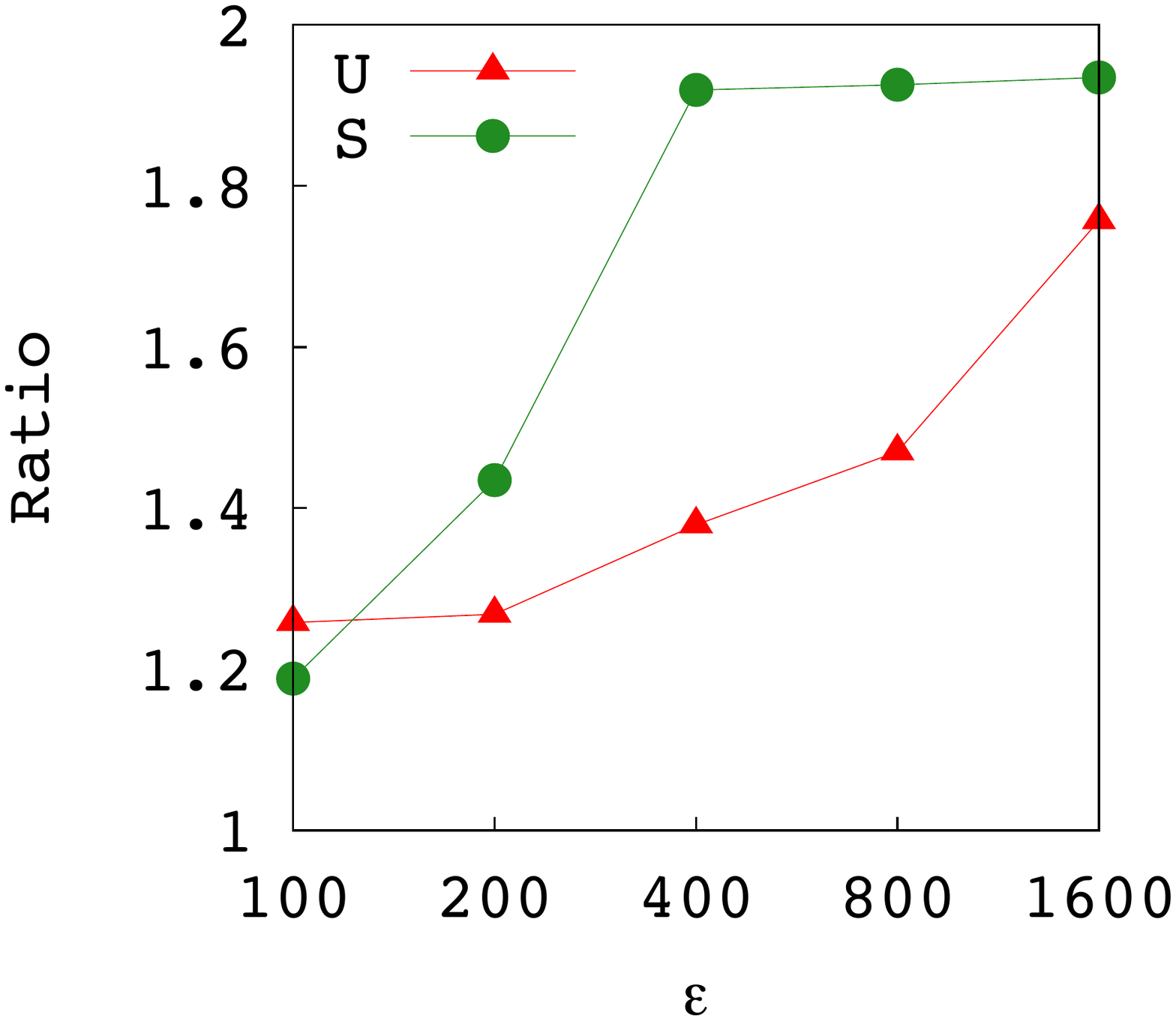}
\vspace{-8pt}
\caption{Index size vs. approximation ratio for varying $\epsilon$}
\vspace{-8pt}
\label{fig:epsilon_ratio}
\end{figure}

\myparagraph{Varying $|W|$}

Figure~\ref{fig:w} and Figure~\ref{fig:w_foursq} show the impact of
varying the number of queries in the sliding window, $|W|$, for {\Auss}
and {\Foursq}, respectively. For {\Auss}, the experiments were conducted using uniform and
skewed query sets, while the {\Foursq} query set is
derived directly from user check-ins.

For both datasets, the number of popularity computations required by
the approximate approach is about $3$ orders of magnitude less than
the baseline.
The reason is two-fold: (i) In the approximate approach, we compute
the popularity of only the objects necessary to update the result.
If the result objects of the previous window are found as valid, we
do not need to compute the popularity of any additional object.
In contrast, the baseline solution must update the popularity for 
all of the objects that satisfy the query constraint.
(ii) Since the popularity function is an average aggregation (see
Sec.~\ref{sec:ps}), the popularity of an object usually does not
change drastically as $|W|$ increases.
Therefore, the result objects in a window are more likely to stay
valid in subsequent windows for larger values of $|W|$, thereby
requiring even fewer objects being checked.
As shown in Figure~\ref{fig:tw} and Figure~\ref{fig:tw_foursq}, fewer
popularity computation directly translates to lower running time.

In {\Auss}, the performance in both uniform and skewed query sets
improves $|W|$ increases, but drops slightly from $|W| =
800$ to $|W| = 1600$ for the approximate approach.
The reason is that, if the results are \emph{not valid} for
a window, we need to look in the validation objects, which is a
subset of the objects that satisfy the constraint of at least one
query in the current window.
So although the results update less often for larger $|W|$, an
update in the results may require checking more objects for a larger
$|W|$.

\myparagraph{Varying query range}

Figure~\ref{fig:r} shows the performance when varying the radius of
each query as a percentage of the dataspace.
We vary the query radius only for {\Auss}, as we use the radius that
covers the check-in locations of a user as the query radius in 
the {\Foursq} dataset.
Here, the number of objects that fall into the query range grows
as query radius increases.
Therefore, the performance of the baseline declines rapidly when the
query radius increases.
In contrast, the approximate approach computes the popularity of only
the objects that can be a result, which is a subset of the objects
that fall within the query range.
Thus, the approximate approach outperforms the baseline, and the
benefit is more significant as the query radius increases.

\myparagraph{Varying $m$}
The experimental results when varying the number of result objects,
$m$, are shown in Figure~\ref{fig:m} and Figure~\ref{fig:m_foursq}
for {\Auss} and {\Foursq}, respectively.
Here, the performance of the baseline does not vary much, as the
baseline computes the popularity for all of the objects that fall
within the query range regardless of the value of $m$.
The approximate approach outperforms the baseline, because the
approximate approach considers only the objects that can potentially
be in the top-$m$ results.
As more objects qualify to be a result, the performance of the
approximate approach decreases with the increase of $m$.

\myparagraph{Varying $\epsilon$}
Figure~\ref{fig:epsilon} shows the performance of the approaches when 
varying the approximation parameter $\epsilon$ for \Auss dataset.
The approximate approach consistently outperforms the baseline for
all choices of $\epsilon$.
As the rank of an object is more accurately approximated for a
smaller value of $\epsilon$, it leads to checking fewer number of
objects and a lower runtime.
As a result, the performance of the approximate approach gradually
decreases with the increase of $\epsilon$.

\myparagraph{Varying $B$}
We vary the block size of the rank lists as the parameter $B$,
and measure the performance.
We find that the number of objects to check does not vary with $B$,
because, if the result of a window needs to be updated, the same set
of validation objects are retrieved regardless of the rank list block size
Therefore, we only show the runtime for varying $B$ in
Figure~\ref{fig:b}.
For each $B$, the total runtime is shown as a breakdown of the
computation time for (i) block-level safe rank, (ii) object-level
safe rank, and (iii) validation object computation for both uniform
and skewed query sets.
From Figure~\ref{fig:b} we can conclude that: (1) as the total number
of blocks decreases for higher $B$, the time required to compute the
block-level safe rank also decreases; and (2) the validation object
lookups dominate the computational costs of the approximate solution.

\subsection{Effectiveness Evaluation}

\myparagraph{Varying $m$}

As shown in Figure~\ref{fig:dataset1}, the query locations originally
follow a skewed distribution, and most of the query locations are
clustered in a small area (which is the central business district of
that city), while the rest of the queries are scattered regionally for \Auss dataset.
In the uniform query set, the queries are repeated uniformly, thus
the upsized query set also follows the same (skewed) distribution of
the original query set.
For this reason, we evaluated our effectiveness as a percentage of
result overlap when using the uniformly upsized query set to capture
a more realistic scenario.

The percentage of result overlap between the top-$m$ approximate
results and the top-$k$ exact results for \Auss dataset are shown in
Figure~\ref{fig:overlap}, where $k$ ranges from $10$ to $200$ and we
set three choices of $m$ ($10$, $50$, $100$).
We find that as $k$ increases, the overlap percentage also increases.
For $m=50$ and $100$, the overlap percentage quickly reaches $90\%$ when
$k=50$.
Note that, if multiple objects have the same popularity value, we
treat their rank position in the result as equivalent.

More experiment results on the percentage of result overlap for the \Foursq dataset and the approximation ratio for both datasets for the varying $m$ can be found in Appendix~\ref{appendix:exp}.



\myparagraph{Varying query range}  Please refer to Appendix~\ref{appendix:exp} for the approximation ratio w.r.t varying query ranges.  


\myparagraph{Varying $\epsilon$, space vs. effectiveness tradeoff}
Figure~\ref{fig:epsilon_ratio} shows the tradeoff between the space requirement and the effectiveness in terms of approximation ratio for varying $\epsilon$. Here, the x-axis represents the index size in GB for both datasets, where $\epsilon$ is varied from $1$ to $5$ at an interval of $1$.  
Since the approximate popularity of an object becomes closer to the
exact popularity as $\epsilon$ decreases, the approximation ratio
also improves for smaller $\epsilon$.

\section{Conclusion} \label{sec:conclusion}
In this paper, we presented the problem of top-$m$ rank aggregation of
spatial objects for streaming queries.
We showed how to bound the rank of an object for any unseen query,
and then proposed an exact solution for the problem.
We then proposed an approximate solution with a guaranteed
error bound, in which used {\em safe ranking}
to determine whether the current result is still valid or not when
new queries arrive, and {\em validation objects}
to limit the number of objects to update in the top-$m$ results.
We conducted a series of experiments on two real datasets, and
show that the approximate approach is about $3$ orders of
magnitude efficient than the exact solution on the collections, and the
results returned by the approximate approach have more than a $90\%$
overlap with the exact solution for $m$ higher than $50$.
Our work combines three important problem domains (rank aggregation,
continuous queries and spatial databases) into a single context.
In future work, we intend to continue exploring other spatial
query constraints and rank aggregation functions using our framework
in order to better understand how the interplay between these three
important domains can be leveraged to solve other cross-disciplinary
problems.

%

\begin{small}
\balance
\setlength{\bibsep}{0.0pt}
\bibliographystyle{abbrvnat}
\bibliography{ref_all}
\end{small}

%
%

\appendix
\section{Proof of approximation error bound} \label{appendix:proof}
 \begin{proof}
Here, $\ar(o,q)$ is the average value of the $\LR(o,c)$ and $\UR(o,c) =
(1+\epsilon) \times \LR(o,c)$, where $c$ is the cell that contains
$q$.
Therefore, the difference between $\rank(o,q)$ and $\ar(o,q)$ is
maximum when $\rank(o,q) = \LR(o,c)$ or $\rank(o,q) = \UR(o,c)$.

From Equation~\ref{eqn:ap}, the difference between the exact and the
approximate popularity computation of an object $o$ is derived from substituting the $\rank(o,q)$ by $\ar(o,q)$ for each query $q$ in $W$.
If $o$ does not satisfy $\constraint(q)$, the contribution to the
popularity for $q$ is $0$ for both cases.
Therefore, the difference between $\rho(o,W)$ and $\ap(o,W)$ is
maximum when either (i) $\rank(o,q_i) = \LR(o,c_i)$, or (ii)
$\rank(o,q_i) = \UR(o,c_i)$ for each $q_i$ in $W$.
We denote $\lambda_i = {\sum_{i=1}^{|W|}} \LR(o,c_i)$ for ease of presentation.

\noi (i) If $\rank(o,q_i) = \LR(o,c_i)$ for each $q_i$ in $W$, then 
\begin{align*}
(1)\D \rho(o,W) &= \frac{\mathlarger{\sum_{i=1}^{|W|}} N-\LR(o,c_i)+1}{|W|}, \mbox{ and } \\ 
(2)\D \ap(o,W) &= \frac{\mathlarger{\sum_{i=1}^{|W|}} N- (1+\epsilon/2) \times \LR(o,c_i) +1}{|W|} \mbox{ (from Eqn.~\ref{eqn:ar}).}
\end{align*}

\begin{align*}
 \rho(o,W) - \ap(o,W)  &=  \frac{ \mathlarger{\sum_{i=1}^{|W|}} - \LR(o,c_i) + (1+\epsilon/2) \LR(o,c_i)}  {\displaystyle |W|} \\
&= \epsilon/2 \times \frac{\displaystyle \lambda_i} {\displaystyle |W|}\\
\frac{\displaystyle \rho(o,W)}{\displaystyle \rho(o,W) - \ap(o,W)} &=  \frac{\mathlarger{\sum_{i=1}^{|W|}} N- \LR(o,c_i)+1} {\displaystyle \epsilon/2 \times \lambda_i}\\
&= \frac{\displaystyle W \times N + W - \lambda_i} {\displaystyle \epsilon/2 \times \lambda_i}
\end{align*}

\noi Here, $\LR(o,c)$ is the lower bound rank estimation, and the
rank of an object is between $[1,N]$, hence, $W \le \lambda_i \le W \times N$.
Therefore, the value of the nominator $W \times N + W - \lambda_i$ is also
between $[W,W \times N]$. 
So by setting the lowest value of $\lambda_i$ in the equation, we get
the following inequality, 

\begin{align*}
\frac{\displaystyle \rho(o,W)}{\displaystyle \rho(o,W) - \ap(o,W)} &\le \frac{\displaystyle W \times N + W - W} {\displaystyle \epsilon/2 \times W} \\
&\le \frac{\displaystyle N} {\displaystyle \epsilon/2 }\\
&\le \frac{\displaystyle 2N} {\displaystyle \epsilon }\\
\Rightarrow \frac{\displaystyle \rho(o,W) - \ap(o,W)}{\displaystyle
\rho(o,W)} &\ge \frac{\displaystyle \epsilon }{\displaystyle 2N },
\mbox{ (by taking the inverse)}\\
\Rightarrow  1- \frac{\displaystyle \ap(o,W)}{\displaystyle \rho(o,W)} &\ge \frac{\displaystyle \epsilon }{\displaystyle 2N }\\
\Rightarrow  \frac{\displaystyle \ap(o,W)}{\displaystyle \rho(o,W)} &\le 1 - \frac{\displaystyle \epsilon }{\displaystyle 2N }
\end{align*}

\noi (ii) If $\rank(o,q_i) = \UR(o,c_i)$ for each $q_i$ in $W$, then 
\begin{align*}
(1)\D \rho(o,W) &= \frac{ \mathlarger{\sum_{i=1}^{|W|}} N- (1+ \epsilon) \times \LR(o,c_i)+1}{|W|} \mbox {, and } \\
(2)\D \ap(o,W)  &= \frac{ \mathlarger{\sum_{i=1}^{|W|}} N- (1+\epsilon/2) \times \LR(o,c_i) +1}{|W|} \mbox{ (from Eqn.~\ref{eqn:ar}). }
\end{align*}

\begin{align*}
 \ap(o,W) - \rho(o,W) &=  \frac{ \mathlarger{\sum_{i=1}^{|W|}} - (1+\epsilon/2) \times \LR(o,c_i) +  (1+ \epsilon) \times \LR(o,c_i)}  {\displaystyle |W|} \\
&= \epsilon/2 \times \frac{\displaystyle \lambda_i}   {\displaystyle |W|}\\
\frac{\displaystyle \ap(o,W)}{\displaystyle \ap(o,W) - \rho(o,W)} &=  \frac{\mathlarger{\sum_{i=1}^{|W|}} N- (1+ \epsilon/2) \times \LR(o,c_i)+1} {\displaystyle \epsilon/2 \times \lambda_i}\\
&= \frac{\displaystyle W \times N + W - (1+ \epsilon/2) \times \lambda_i} {\displaystyle \epsilon/2 \times \lambda_i}
\end{align*}

Setting the lowest value of $\lambda_i = W$ in the equation produces the
following inequality,
\begin{align*}
\frac{\displaystyle \ap(o,W)}{\displaystyle \ap(o,W) - \rho(o,W)} &\le \frac{\displaystyle W \times N + W - (1+ \epsilon/2) \times W} {\displaystyle \epsilon/2 \times W}\\
&\le \frac{\displaystyle N - \epsilon/2} {\displaystyle \epsilon/2 }\\
\Rightarrow  \frac{\displaystyle \ap(o,W) - \rho(o,W)}{\displaystyle \ap(o,W)} &\ge  \frac {\displaystyle \epsilon/2 }{\displaystyle N - \epsilon/2}, \mbox{ (by taking the inverse)}
\end{align*}

\noi Since the value $\epsilon$ is between $[0,N-1]$, the denominator $N - \epsilon/2$ can be a maximum of $N$. 
\begin{align*}
1- \frac{\displaystyle  \rho(o,W)}{\displaystyle \ap(o,W)} &\ge \frac {\displaystyle \epsilon/2 }{\displaystyle N - \epsilon/2} \ge \frac {\displaystyle \epsilon/2 }{\displaystyle N}\\
\Rightarrow \frac{\displaystyle  \rho(o,W)}{\displaystyle \ap(o,W)}  &\le 1-  \frac{\displaystyle \epsilon }{\displaystyle 2N }
\end{align*}
\end{proof}

\section{Additional experiment results} \label{appendix:exp}
In this section, we show additional experiment results on our effectiveness study. 
Recall the experiment setting, although the primary purpose of the real country-level \Foursq is to test out the scalability of our approximate and exact solution on real data\footnote{Note that, in reality one seldom issues a spatial range query while the candidates are objects spread over the whole big country}, we also report its effectiveness results for the completeness of experiments.

\myparagraph{Varying $|W|$}
Table~\ref{table:ratio_w} shows the average approximation ratio for
both datasets.
Although the average approximation ratio gradually improves for both
uniform and skewed query sets as $|W|$ increases, the change
does not follow any obvious pattern.
The explanation for this random behaviour is that popularity is an
average aggregation of $|W|$ ranks, so if both the exact and the
approximate popularity do not change at the same rate with $|W|$,
their ratios do not change in a fixed way.

\myparagraph{Varying $m$}
Figure~\ref{fig:overlap_foursq} shows the percentage of overlap between the top-$m$ approximate results and the top-$k$ exact results for {\Foursq} dataset. Although the overlap becomes close to 100\% for higher $m$, the overlap is not as good as the \Auss dataset for lower values of $m$. The reason is as follows. As shown in Figure~\ref{fig:dataset} the objects in \Foursq are clustered into cities, and the cities are scattered in different parts of the USA. On the other hand, the query locations are distributed all over the dataspace, as a user can check-in at different cities. Therefore, the popularity values of most of the objects in a city are very close to each other. Figure~\ref{fig:screenshot} shows a screenshot of the top-$10$ popularities computed in the baseline approach at three example instances. As we can see, the final rank of two objects can be very far away for a slight difference in their popularity values; for example, in the first example instance, the difference between every adjacent objects' popularity score is only 0.25 in average while the absolute values are at the scale of 50K.

\begin{table}[h]
\centering
\caption{Approximation ratio for varying $|W|$}
\label{table:ratio_w}
\begin{small}
\begin{tabular}{|c|c|c|c|c|c|c|}
\hline

\multicolumn{2}{|c|}{\backslashbox{Dataset}{$|W|$}} & 100 & 200 & 400 & 800 & 1600\\
\hline
\multirow{2}{*}{\Auss} & $U$ & 2.12 & 1.60 & 1.57 & 1.67 & 1.55\\
\cline{2-7}
& $S$ & 3.19 & 1.55 & 2.14 & 1.30 & 1.34\\
\hline
\multicolumn{2}{|c|}{\Foursq}  & 2.76 & 6.87 & 3.49  & 3.15 & 2.33  \\

\hline

\end{tabular}
\end{small}
\end{table}

\begin{table}[h]
\centering
\caption{Approximation ratio for varying $m$}
\label{table:ratio_m}
\begin{small}
\begin{tabular}{|c|c|c|c|c|c|c|c|}
\hline

\multicolumn{2}{|c|}{\backslashbox{Dataset}{$m$}} & 1 & 5 & 10 & 20 & 50& 100\\
\hline
\multirow{2}{*}{\Auss} & $U$ & 3.00 & 4.79 & 1.57 & 1.56 & 1.49 & 1.49\\
\cline{2-7}
& $S$ & 5.61 & 2.03 & 2.14 & 1.60 & 1.17 & 1.16\\
\hline
\multicolumn{2}{|c|}{\Foursq}  &  1.57 & 2.62 & 2.68  & 3.31 & 3.37 & 2.47 \\

\hline

\end{tabular}
\end{small}
\end{table}

\begin{figure}[h]
\centering
\includegraphics[width=0.25\textwidth]{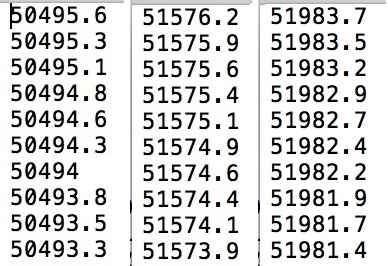}
\caption{Popularity values of top-$10$ objects in \Foursq dataset}
\label{fig:screenshot}
\end{figure}

\begin{figure}[h]
\centering
\includegraphics[angle =-90,trim = 20mm 28mm 20mm 20mm, clip,width=0.40\textwidth]{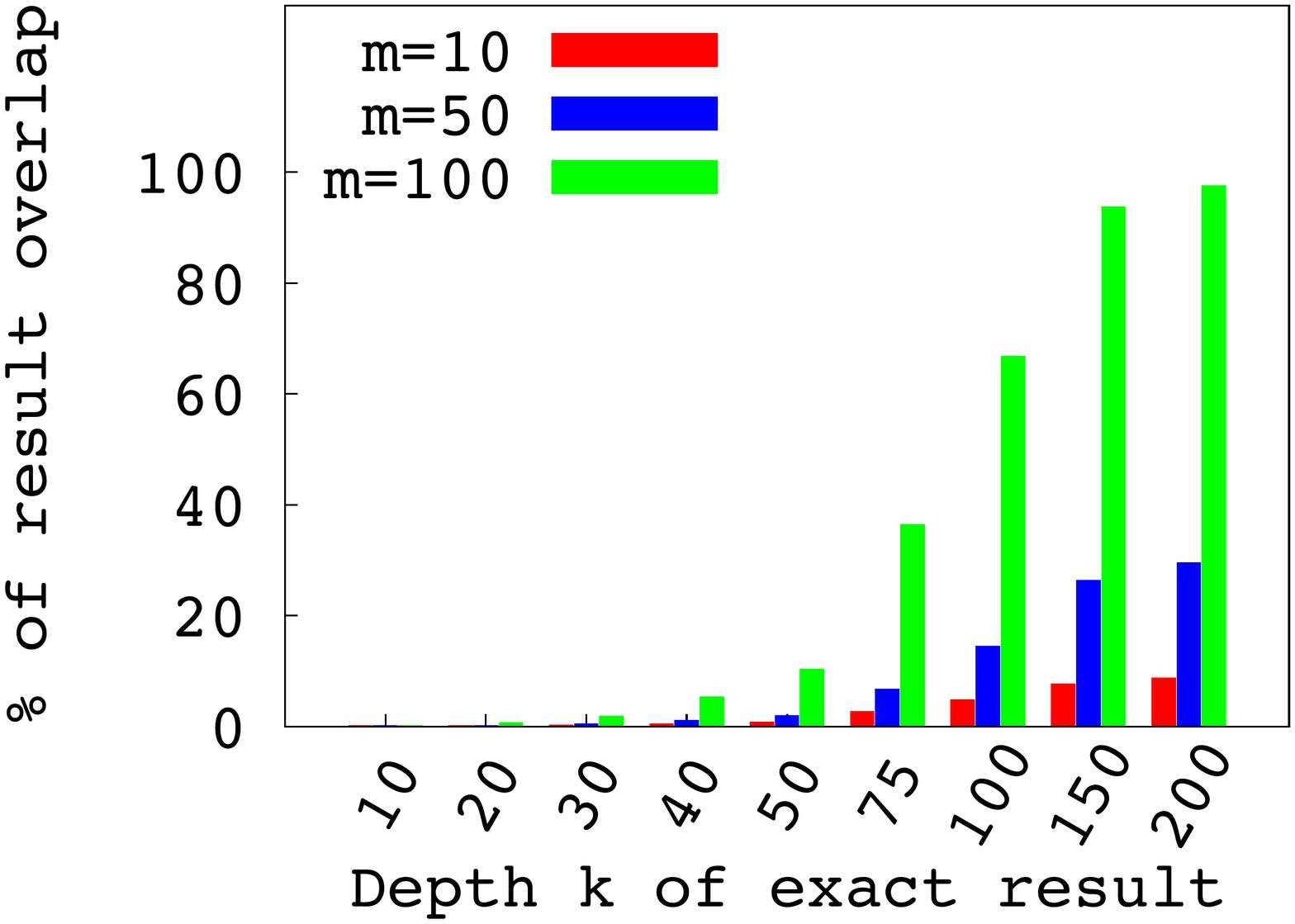}
\caption{\% of result overlap for varying $m$ in \Foursq dataset}
\label{fig:overlap_foursq}
\end{figure}


Table~\ref{table:ratio_m} shows the approximation ratio for varying $m$ for \Auss dataset. As shown in the table, the approximation ratio keeps improving with the increase of $m$, probably because most objects in the top-$m$ ranked list have very similar scores in both their approximate popularity and approximate popularity when $m<100$. the top-$m$ ranked list are very close to each have very similar exact and approximate popularity scores.

\myparagraph{Approximation ratio for varying query range}
\begin{table}[h]
\centering
\caption{Approximation ratio for varying query radius}
\label{table:ratio_r}
\begin{small}
\begin{tabular}{|c|c|c|c|c|c|c|}
\hline

\multicolumn{2}{|c|}{\backslashbox{Dataset}{Query radius}} & 1 & 2 & 4 & 8 & 16\\
\hline
\multirow{2}{*}{\Auss} & $U$ & 2.55 & 1.55 & 1.57 & 1.61 & 1.63\\
\cline{2-7}
& $S$ & 3.32 & 2.88 & 2.14 & 2.39 & 3.55\\

\hline

\end{tabular}
\end{small}
\end{table} 

The approximation ratio of the results w.r.t. varying query ranges are shown in Table~\ref{table:ratio_r}. We find that the approximation ratio does not indicate any significant pattern for this parameter, because the approximation calculation does not depend on the query radius or the number of5objects falling within that range.

\balancecolumns

\end{document}